\documentclass[12pt]{article}
\usepackage[english]{babel}
\usepackage[utf8]{inputenc}
\usepackage[T1]{fontenc}
\usepackage{amsbsy}
\usepackage{amsfonts}
\usepackage{amsmath}
\usepackage{amssymb}
\usepackage{amsthm}
\usepackage{graphicx} 
\usepackage[pdftex]{color}
\usepackage{enumerate}
\usepackage{float}
\usepackage{geometry, calc, color}
\usepackage{indentfirst}
\usepackage{longtable}
\usepackage{multirow}
\usepackage{natbib}
\usepackage{xr-hyper}
\usepackage{hyperref}
\usepackage{caption}
\usepackage{subcaption}
\usepackage{booktabs}
 \usepackage{multirow}
\usepackage{xcolor}
\usepackage{comment}
\usepackage[linesnumbered,ruled,vlined]{algorithm2e}

\usepackage{lmodern} 

\newtheorem{Remark}{Remark}[section]
\newtheorem{Lemma}{Lemma}[section]

\newtheorem{Definition}{Definition}[section]
\newtheorem{Corollary}{Corollary}[section]
\newtheorem{Theorem}{Theorem}[section]
\newtheorem{Proposition}{Proposition}[section]

\usepackage[textwidth=3cm, textsize=footnotesize]{todonotes} 

\hypersetup{
  colorlinks=true,
  linkcolor=red,
  citecolor=teal,
  filecolor=blue,
  urlcolor=blue,
}
\newcommand{\blind}{1}

\SetCommentSty{mycommfont}

\usepackage{anysize}
\marginsize{2cm}{2.5cm}{2.5cm}{2cm}

\begin{document}

\def\spacingset#1{\renewcommand{\baselinestretch}%
	{#1}\small\normalsize} \spacingset{1}

\if1\blind
{
	\title{\bf Time-Varying Dispersion Integer-Valued \\ GARCH Models\footnote{Paper published in the Journal of Time Series Analysis. \\Open Access available at \url{https://onlinelibrary.wiley.com/doi/10.1111/jtsa.12838}}}
	 \author{Wagner Barreto-Souza$^\sharp$\footnote{Email: \textcolor{teal}{\texttt{wagner.barreto-souza@ucd.ie}}},\,\, Luiza S.C. Piancastelli$^\sharp$\footnote{Email: \textcolor{teal}{\texttt{luiza.piancastelli@ucd.ie}}},\,\, Konstantinos Fokianos$^\flat$\footnote{Email: \textcolor{teal}{\texttt{fokianos@ucy.ac.cy}}}\,\\ and Hernando Ombao$^\star$\footnote{Email: \textcolor{teal}{\texttt{hernando.ombao@kaust.edu.sa}}}\hspace{.2cm}\\
		{\normalsize \it $^\sharp$School of Mathematics and Statistics, University College Dublin, Dublin 4, Republic of Ireland}\\
			{\normalsize \it $^\flat$Department of Mathematics and Statistics, University of Cyprus, Nicosia, Cyprus}\\
		{\normalsize \it $^\star$Statistics Program, King Abdullah University of Science and Technology, Thuwal, Saudi Arabia}}
	\maketitle
} \fi
\if0\blind
{
	\bigskip
	\bigskip
	\bigskip
	\begin{center}
		{\LARGE\bf Time-Varying Dispersion Integer-Valued GARCH Models}
	\end{center}
	\medskip
} \fi

\begin{abstract}
We introduce a general class of INteger-valued Generalized AutoRegressive Conditionally Heteroscedastic (INGARCH) processes by allowing simultaneously time-varying mean and dispersion parameters. We call such models time-varying dispersion INGARCH (tv-DINGARCH) models. More specifically, we consider mixed Poisson INGARCH models and allow for dynamic modeling of both mean  and dispersion parameters. We derive conditions to obtain stochastic stability of tv-DINGARCH processes. Additionally, we study maximum likelihood estimation in detail including its asymptotic distribution. A restricted bootstrap procedure is proposed for testing constant dispersion against time-varying dispersion. Monte Carlo simulation studies are presented for checking point estimation, standard errors, and the performance of the restricted bootstrap approach. We apply the tv-DINGARCH process to model the weekly number of reported measles infections in North Rhine-Westphalia, Germany, from January 2001 to May 2013, and compare its performance to the ordinary INGARCH approach.
\end{abstract}

\noindent {\it \textbf{MOS subject classifications}:} 62M09; 62M10.

\section{Introduction}
  
Modeling count time series data is a challenging and exciting research topic with applications in many areas such as epidemiology, sociology, economics, and health science. A well-established methodology, for dealing with count time series data, has been developed under the framework of the so-called INteger-valued Generalized AutoRegressive Conditional Heteroscedastic (INGARCH) models, which have been initially studied and explored by \cite{hei2003}, \cite{ferlandetal2006}, \cite{foketal2009}, and \cite{fok2011}. The INGARCH nomenclature emerges from the fact that for a Poisson distribution (this assumption is imposed in the aforementioned papers), the mean equals its variance. Therefore modeling of mean implies modeling of variance like in the classic  GARCH models introduced by \cite{bollerslev1986}. Consequently, such models are considered in the literature as an integer counterpart of the GARCH models although this terminology should be used cautiously because ordinary GARCH processes are employed for modeling time-varying variances.

A Poisson $\mbox{INGARCH}(p,q)$ (with $p,q\in\mathbb N$) model $\{Y_t\}_{t\in\mathbb Z}$ is defined by 
\begin{eqnarray}\label{pingarch}
	Y_t|\mathcal F_{t-1}\sim\mbox{Poisson}(\lambda_t),\quad
	\lambda_t\equiv E(Y_t|\mathcal F_{t-1})=\beta_0+\sum_{i=1}^p\beta_i Y_{t-i}+\sum_{j=1}^q\alpha_j\lambda_{t-j},
\end{eqnarray}	 
where $\mathcal F_{t}=\sigma\{Y_{s}, s \leq t\}$, $\beta_0>0$, $\beta_i\geq0$ and $\alpha_j\geq0$, for all $i=1,\ldots,p$ and $j=1,\ldots,q$. Since $E(Y_t|\mathcal F_{t-1})=\mbox{Var}(Y_t|\mathcal F_{t-1})=\lambda_t$, the model specifies dynamic processes for both the conditional mean and variance, as previously stated.

The Poisson INGARCH model has been extensively considered in the literature. However, due to limitations for fitting adequately real count time series data (e.g. not capturing all the sources of overdispersion although the model itself is overdispersed), several variants of the Poisson INGARCH model have been proposed. These include the negative binomial (NB) INGARCH process by \cite{zhu2011} and \cite{chrfok2014} or the more general mixed Poisson INGARCH models by \cite{chrfok2015}, \cite{silbar2019}, among others. Zero-inflated versions of the Poisson and negative binomial INGARCH models have been studied by \cite{zhu2012a}, while processes dealing with both overdispersion and underdispersion were proposed by \citet{zhu2012b,zhu2012c} and \cite{xuetal2012}. A negative binomial quasi-likelihood based inference for general integer-valued time series models with application to the Poisson and NB INGARCH models has been developed by \cite{aknetal2018}. 

An alternative process to the linear model given in \eqref{pingarch}, was introduced by \cite{foktjo2011} through the log-linear INGARCH processes. These models cope with both negative and positive autocorrelation functions and allow the inclusion of covariates in a straightforward way.  More recently \cite{weietal2022} considered a softplus link function  instead of a logarithmic one.

Based on the discussion so far, the Poisson and negative binomial distributions are frequently used in applications for count time series analysis. In several cases,
the negative binomial model fits more adequately because this distribution depends on a dispersion parameter which facilitates 
flexible modeling. The dispersion parameter is usually assumed fixed and is estimated by using the Pearson residuals; see \cite{chrfok2014} for more. 
The main goal of this paper is to propose a novel and general class of INGARCH processes based on the mixed Poisson distributions (the negative binomial distribution belongs to this class) by {\it allowing both time-varying mean and dispersion parameters}. We call these models time-varying dispersion INGARCH (tv-DINGARCH) models.   Hence, we develop methodology and introduce new models for integer-valued time series, where the assumption of constant dispersion might not hold. The usefulness of the tv-DINGARCH processes over the ordinary INGARCH models will be illustrated in Section \ref{sec:application}, where we show that such processes fit better real data whose stylistic facts cannot be dealt with the existing INGARCH models. 

For example, in a count regression context, \cite{barsim16} demonstrated through a data application on the attendance behavior of high school juniors that a constant dispersion assumption can be violated. Generalized linear models allowing regression structures for both mean and dispersion/precision parameters have been discussed previously 
by \cite{efr1986} and \cite{smy1989}, among others. In this work, we further develop this line of research by studying in detail the joint modeling of mean and dispersion parameters in the context of dependent data. We put special emphasis on the case of INGARCH models for count time series and demonstrate the usefulness of such methodological 
development.  Besides extending the traditional INGARCH models, the proposed class includes the case of constant mean and time-dependent variance, which is a new model class that does not fall in the ordinary INGARCH model approach. Another important feature is volatility modeling, which is a well-explored topic for the case of continuous-valued time series but has been neglected in the context of integer-valued time series. Although the traditional Poisson INGARCH processes consider a time-dependent conditional variance, this in turn is solely driven by the dynamics of the mean process. The main disadvantage of this approach is that imposes severe restrictions on the volatility modelling as illustrated in the real data application in this paper. The class  tv-DINGARCH models relaxes this assumption by considering a time-dependent dispersion parameter, therefore also controlling the conditional variance and allowing for additional sources of volatility.

In the context of time-varying models for INGARCH processes, \cite{roykar2021} introduced a Poisson INGARCH model with time-varying coefficients and  \cite{ratsam2023} study a time-varying INGARCH model based on the zero-inflated Poisson distribution. These authors consider dynamics for the inflation parameter which is assumed to depend on exogenous variables.
Furthermore, \citet{doukhanetal2022} study in detail first order non-stationary INGARCH models and prove mixing conditions
under natural assumptions. Our approach is different in several aspects. First, we consider mixed Poisson distributions (rather than Poisson and zero-inflated Poisson models) which cover a broad class of count distributions that include the Poisson and negative binomial.  In addition, we consider dynamics that enter through the dispersion parameter. One advantage of the latter feature is that we are able to establish desirable stability properties for the models we consider, such as stationarity and ergodicity, by employing e-chain theory; see \cite{meytwe1993}.

The paper is organized as follows. In Section \ref{sec:bingarch}, we define the class of tv-DINGARCH models and then derive its stochastic properties. We establish conditions ensuring stationarity and ergodicity of the count processes. Section \ref{sec:stat_inf} is devoted to the estimation of the parameters and the derivation of their associated asymptotic properties as well. Aiming at testing constant dispersion in practice, a restricted bootstrap procedure is proposed in Section \ref{sec:test}. Monte Carlo simulation studies are presented for checking point estimation, standard errors, and the performance of the restricted bootstrap approach. In Section \ref{sec:application},  we apply the tv-DINGARCH process to model the weekly number of reported measles infections in North Rhine-Westphalia, Germany, from January 2001 to May 2013, and compare its performance to the ordinary INGARCH approach. Some technical results and proofs are provided in the Appendix. This paper contains Supplementary Material.

\paragraph{Notation:} 
We say that a random variable $Y$ follows a mixed Poisson (MP) distribution if satisfies the stochastic representation $Y|Z=z\sim\mbox{Poisson}(\lambda z)$, with $Z$ following some non-negative distribution with $E(Z)=1$ (standardization) and $\mbox{Var}(Z)=\phi^{-1}$, for $\lambda,\phi>0$. In this case, $E(Y)=\lambda$ and $\mbox{Var}(Y)=\lambda+\lambda^2/\phi$. We denote $Y\sim\mbox{MP}(\lambda,\phi;\mathcal D)$, where $\mathcal D$ denotes the distribution of $Z$. A random variable $Z$ following a Gamma distribution, with respective shape and scale parameters $a>0$ and $b>0$, is denoted by $Z\sim\mbox{Gamma}(a,b)$, where $E(Z)=a/b$ and $\mbox{Var}(Z)=a/b^2$. If the mixing variable  $Z\sim\mbox{Gamma}(\phi,\phi)$ (i.e. $\mathcal D=\mbox{Gamma}$), then we obtain that $Y$ follows a negative binomial (NB) distribution, i.e.  $Y\sim\mbox{NB}(\lambda,\phi)$. 

For an $l$-dimensional vector ${\bf x}=(x_1,\ldots,x_l)^\top$, let $\|{\bf x}\|_p=(\sum_{i=1}^l|x_i|^p)^{1/p}$, for $p\in[1,\infty)$, and $\|{\bf x}\|_p=\max_{1\leq i\leq l}|x_i|$ for $p=\infty$. The induced $p$-norm for a $m\times l$ matrix ${\bf C}$  is then defined by $\|{\bf C}\|_p=\max_{{\bf x}\neq{\bf0}}\{\|{\bf C x}\|_p/\|{\bf x}\|_p:\, {\bf x}\in\mathbb R^l\}$, for $p\in[1,\infty]$. For a square matrix {\bf A}, we denote its spectral radius by $\rho({\bf A})$.

\section{The tv-DINGARCH Models}\label{sec:bingarch}

In this section, we propose a class of time-varying dispersion INGARCH models by allowing both mean and dispersion parameters to depend on the lagged values of the observed series and their past values as follows.

\begin{Definition}\label{def:dingarch} (tv-DINGARCH processes)
	A $\mbox{tv-DINGARCH}(p_1,p_2,q_1,q_2)$ process $\{Y_t\}_{t\in\mathbb Z}$ is defined by $Y_t|\mathcal F_{t-1}\sim\mbox{MP}(\lambda_t,\phi_t;\mathcal D)$, with
	\begin{eqnarray}\label{DINGARCH_dyn}
		\lambda_t=f(Y_{t-1},\ldots,Y_{t-p_1};\lambda_{t-1},\ldots,\lambda_{t-q_1}),\quad
		\phi_t=g(Y_{t-1},\ldots,Y_{t-p_2};\phi_{t-1},\ldots,\phi_{t-q_2}),
	\end{eqnarray}	 
	where $\mathcal F_{t}=\sigma\{Y_{s}, s \leq t\}$ and $f:\mathbb N^{p_1}\times(0,\infty)^{q_1}\rightarrow(0,\infty)$, and $g:\mathbb N^{p_2}\times(0,\infty)^{q_2}\rightarrow(0,\infty)$.
\end{Definition}

Of particular interest to our study and to diverse applications is the negative binomial (NB) model when it is imposed in Definition \ref{def:dingarch}. Then, the  conditional probability function of $Y_t$ given $\mathcal F_{t-1}$  is given by
\begin{eqnarray*}
	P(Y_t=y|\mathcal F_{t-1})=\dfrac{\Gamma(y+\phi_t)}{y!\Gamma(\phi_t)}\left(\dfrac{\lambda_t}{\lambda_t+\phi_t}\right)^y\left(\dfrac{\phi_t}{\lambda_t+\phi_t}\right)^{\phi_t},\quad y\in\mathbb N_0.
\end{eqnarray*}	
The negative binomial model will be assumed to prove the finiteness of moments and to derive the asymptotic theory for conditional maximum likelihood estimators in what follows. 

Definition \ref{def:dingarch} is general and some additional assumptions on the functions $f$ and $g$ are necessary to study in detail these processes.  We consider the first-order model  ($p_1=p_2=q_1=q_2=1$) linear tv-DINGARCH processes, which will be defined in what follows.  This particular linear parametric form is a common choice considered in the literature but also enables a thorough study of stability properties of processes given  by Definition \ref{def:dingarch}.

\begin{Definition}\label{def:lin_dingarch} (Linear tv-DINGARCH processes)
	A linear $\mbox{tv-DINGARCH}(1,1,1,1)$ process $\{Y_t\}_{t \in \mathbb{Z}}$ is given as in Definition \ref{def:dingarch} by assuming that with $f(\cdot)$ and $g(\cdot)$ are linear parametric functions. In other words, assume that 
	\begin{eqnarray}\label{DINGARCH}
		\lambda_t=\beta_0+\beta_1 Y_{t-1}+\beta_2\lambda_{t-1},\quad
		\phi_t=\alpha_0+\alpha_1 Y_{t-1}+\alpha_2\phi_{t-1},
	\end{eqnarray}	 
	where $\beta_0,\alpha_0>0$ and $\beta_1,\beta_2,\alpha_1,\alpha_2\geq0$. 
\end{Definition}

\begin{Remark} \rm 
	We  use  $Y_{t-1}$ to model $\phi_t$ in eq. \eqref{DINGARCH}, and more generally in Def. \ref{def:dingarch}, instead of terms such as $Y^2_{t-1}$ and $|Y_{t-1}|$ considered for the continuous GARCH models and their modification. This is so because $Y_{t-1}$ is non-negative as we deal with count data.  
\end{Remark}	

\begin{Remark}\label{pure_ingarch} \rm 	 Some particular cases are obtained from the  tv-DINGARCH class defined above. When  $\alpha_1=\alpha_2=0$, we obtain the ordinary linear mixed Poisson INGARCH \citep{chrfok2015,silbar2019} models as particular cases. Additionally, by taking $\alpha_0\rightarrow0^+$, we also obtain the Poisson INGARCH model as a limiting member of our proposed class. Another interesting and novel model arises when $\beta_1=\beta_2=0$. Under this setting, the mean of the INGARCH process is constant and the variance is time-dependent as in the case of ordinary GARCH models \citep{bollerslev1986}. Such property does not hold for ordinary Poisson INGARCH  models  given by Equation (\ref{pingarch}).
	When $\beta_1=\beta_2=0$	 we call  \eqref{DINGARCH}  Pure INGARCH (P-INGARCH) process. We use the term ``pure" to emphasize the fact that the  model mimics the traditional continuous GARCH models (constant mean and time-varying variance). Hence, our approach is general and encompasses many different models studied earlier in the literature.
\end{Remark}

Simulated trajectories of the linear tv-DINGARCH models for a few parameter settings are shown in the Supplementary Material. 

\subsection{Stationarity and Ergodicity}\label{subsec:properties}

We now study the stochastic properties of the $\mbox{tv-DINGARCH}(1,1,1,1)$ models.  Linearity is a common assumption in the literature as discussed before, which is justified due to successful empirical applications. Conditions for the existence and stationarity of the process \eqref{DINGARCH} can be established, for example, by following the strategy by \cite{chrfok2014}, which relies on establishing weak dependence (see \cite{douwin2008}). In this work, we prove such properties based on the
theory of  e-chains \citet{meytwe1993}, as follows.

The dynamic latent processes $\{\lambda_t\}_{t\in \mathbb{Z}}$ and $\{\phi_t\}_{t\in \mathbb{Z}}$ given in (\ref{DINGARCH}) can be written compactly  in a vector form. By defining  $\boldsymbol\xi_t=(\lambda_t,\phi_t)^\top$, $\boldsymbol\tau=(\beta_0,\alpha_0)^\top$,
${\bf A}=\begin{pmatrix} 
	\beta_2 & 0\\
	0 & \alpha_2 
\end{pmatrix}$
and
${\bf B}=\begin{pmatrix}
	\beta_1 & 0\\
	0 & \alpha_1 
\end{pmatrix}$,
we obtain  that $\boldsymbol\xi_t=\boldsymbol\tau+{\bf B}(Y_{t-1},Y_{t-1})^\top+{\bf A}\boldsymbol\xi_{t-1}$. 
Using  this  representation, we  adapt the strategy in \cite{liu2012} to establish  existence and stationarity of 
\eqref{DINGARCH}. In that work, the author provided stochastic properties of a bivariate Poisson INGARCH model. The key point is to show that $\{\boldsymbol\xi_t\}_{t\in\mathbb Z}$ is an e-chain; see  \citet[Ch.18]{meytwe1993}. In addition, the property of Geometric Moment Contraction (GMC), as studied by \cite{diafree1999} and \cite{wusha2004}, enables us to prove the desired ergodicity and uniqueness of the stationary distribution of $\{\boldsymbol\xi_t\}_{t\in \mathbb{Z}}$.

	\begin{Remark} \rm
		In what follows, we prove all theoretical results for the negative binomial linear tv-DINGARCH models as given in  \ref{def:lin_dingarch} except Thm. \ref{Thm_StatErg} and Corollary \ref{Corollary:PQmodel} which hold true for any linear tv-DINGARCH mixed Poisson model including the negative binomial model.
\end{Remark}	

\begin{Theorem}\label{Thm_StatErg}
	Let $\{Y_t\}_{t\in\mathbb Z}$ be a tv-DINGARCH process (with fixed mixing distribution $\mathcal D$) as in \eqref{DINGARCH}. \\
	(a) If $\rho({\bf A}+{\bf B})<1$, then there exists at least one stationary distribution to  $\{\boldsymbol\xi_t\}_{t\in \mathbb{Z}}$, where $\rho(\cdot)$ denotes the spectral radius. In addition, if $\|{\bf A}\|_p<1$ for some $1\leq p\leq\infty$, then the distribution is unique.\\
	(b) If $\|{\bf A}\|_p+2^{1-1/p}\|{\bf B}\|_p<1$ for some $p\in[1,\infty]$, then $\{\boldsymbol\xi_t\}_{t\in \mathbb{Z}}$ is a GMC Markov chain with a unique stationary and ergodic distribution.
\end{Theorem}

The proof is given in the Appendix. 
The importance of Theorem \ref{Thm_StatErg} for data analysis, and more generally, for  modeling count time series is that these conditions are sufficient to have consistency and asymptotic normality of the conditional maximum likelihood estimators, as will be addressed in the next section.
In addition, Theorem \ref{Thm_StatErg} can be extended to the higher-order linear  $\mbox{tv-DINGARCH}(p_1,p_2,q_1,q_2)$ processes, that is, consider  Def. \ref{def:dingarch} with $f$ and $g$ given by
\begin{align} \label{eq:lingarchgeneral}
		\lambda_t=\beta_0+ \sum_{i=1}^{p_{1}}\beta_{1i} Y_{t-i}+  \sum_{j=1}^{p_{2}}\beta_{2j}\lambda_{t-j},\quad
	    \phi_t=\alpha_0+    \sum_{i=1}^{q_{1}}\alpha_{1i} Y_{t-i}+ \sum_{j=1}^{q_{2}}\alpha_{2j} \phi_{t-j}.
\end{align}
If we denote by $\boldsymbol\xi_t=(\lambda_t,\phi_t)^\top=\boldsymbol\tau+\sum_{j=1}^q{\bf B}_j(Y_{t-j},Y_{t-j})^\top+\sum_{i=1}^m{\bf A}_i\boldsymbol\xi_{t-i}$, where $\{{\bf A}_i\}_{i=1}^m$ and $\{{\bf B}_j\}_{j=1}^q$ are defined by 
$$
{\bf A}_{j}=\begin{pmatrix} 
	\beta_{2j} & 0\\
	0 & \alpha_{2j} 
\end{pmatrix}
~\mbox{and} ~
{\bf B}_{j}=\begin{pmatrix}
	\beta_{1j} & 0\\
	0 & \alpha_{1j} 
\end{pmatrix},
$$
$m=\max(p_1,p_2)$, and $q=\max(q_1,q_2)$, then, following the same steps as in the proof of Theorem \ref{Thm_StatErg},
the next result can be established; see \citet[Prop. 4.3.1]{liu2012}

\begin{Corollary} \label{Corollary:PQmodel}
	Consider the linear $\mbox{tv-DINGARCH}(p_1,p_2,q_1,q_2)$ processes \eqref{eq:lingarchgeneral}. Put $m=\max(p_1,p_2)$, and $q=\max(q_1,q_2)$. Then, with the same notation as in Thm. \ref{Thm_StatErg}, if  $\sum_{i=1}^m\|{\bf A}_i\|_p+2^{1-1/p}\sum_{j=1}^q\|{\bf B}_j\|_p<1$, for some $p\in[1,\infty]$, then $\{(Y_t,\xi^\top_{t})\}_{t\in\mathbb Z}$ is a stationary and ergodic process. 
\end{Corollary}

We close this section by proving the finiteness of the moments  (see Appendix for a proof) of the first-order negative binomial tv-DINGARCH process. 
\begin{Theorem}\label{moments_P-INGARCH}
	If  $\|{\bf A}\|_1+\|{\bf B}\|_1<1$,  then the first-order negative binomial tv-DINGARCH process $\{Y_t\}_{t\in\mathbb Z}$ given by \eqref{DINGARCH}  satisfies that $E|Y_t|^k < \infty$,  $E|\lambda_t|^k < \infty$  and  $E|\phi_t|^k < \infty$  for 
	any $k\in\mathbb N$. 
\end{Theorem}

	\begin{Remark}
		We conjecture that Theorem \ref{moments_P-INGARCH} can be extended to hold for any linear mixed Poisson tv-DINGARCH model by combining our arguments and the results of  \citet[Subsection 3.2]{karxek2005} concerning moments of mixed Poisson distributions. This is out of the scope of this work since the focus is on the negative binomial model. 
\end{Remark}

\section{Estimation and Asymptotic Theory}\label{sec:stat_inf}

In this section, we study conditional maximum likelihood estimation and provide some numerical experiments for the linear $\mbox{tv-DINGARCH}(1,1,1,1)$ process in Definition \ref{def:lin_dingarch} under the assumption of negative binomial conditional distribution. Furthermore, asymptotic results will be established for this model. These results can be generalized to the case of general order linear models but we restrict our attention to the first-order case (\ref{DINGARCH}) for ease of presentation.

\begin{Remark} \rm 
Some earlier work, in the context of negative-binomial quasi-likelihood inference  for count time series models was reported by \cite{aknetal2018}. More precisely these authors studied model \eqref{DINGARCH} but with a constant dispersion parameter. 
There are several differences between our work and that of \citet{aknetal2018}. 
	\begin{enumerate}
		\item The first major difference is the type of models considered for the dispersion parameter. We assume a dynamic model--recall equation \eqref{DINGARCH}--whereby the choice $\alpha_1=\alpha_2=0$ implies the model considered by Aknouche et al. (2018). So, in this sense, we extend this contribution.
		\item The estimation method considered by  \cite{aknetal2018} is  different from the direct maximization problem we are studying in this work. Indeed, \cite{aknetal2018} consider a two-stage NB-QMLE (see their Sec. 3.2). We have developed theory for direct maximization of the negative binomial quasi-likelihood under an extended model. 
		\item In \cite{aknetal2018} the invertibility of the Fisher information matrix is not proved but instead it is assumed; see Assumption (A9) (existence, finiteness, and invertibility of $J_r$ following the notation of their paper). 
		The invertibility of the information matrix is addressed by our contribution, in a more general form, by the proof of Theorem 3.1.
		\item Furthermore, these authors assume finiteness of moments \cite[Assumption A3]{aknetal2018}, which requires the existence of $\delta>1$ such that $E(X_t^\delta)<\infty$.  These are natural conditions to be assumed since the authors are studying exclusively quasi-likelihood inference. 
		Instead we have proved such condition  for this extended model we consider. There are additional differences, see for instance \cite[Assumption A2]{aknetal2018} and compare with our Lemma \ref{lemma:starting value}. 
	\end{enumerate}	
\end{Remark}

\subsection{Likelihood Inference}\label{subsec:estimation}

Denote by $\boldsymbol{\theta}_{1}=(\beta_0,\beta_1,\beta_2)^\top$, $\boldsymbol{\theta}_{2}=(\alpha_0,\alpha_1,\alpha_2)^\top$, and 
$\boldsymbol{\theta}=(\boldsymbol{\theta}_{1}^\top, \boldsymbol{\theta}_{2}^\top)^\top=(\theta_{1}, \ldots, \theta_{6})^\top$ the parameter vector
and suppose that $Y_1,\ldots,Y_n$ is a sample  from  a NB tv-DINGARCH process $\{Y_t\}_{t=1}^n$. The conditional log-likelihood function of $Y_2,\ldots,Y_n$ given $Y_1=y_1$ is given  by (up to a constant) 
\begin{equation}
	\begin{aligned}  \label{loglikelihood} 
		\ell(\boldsymbol{\theta}) &  =  \sum_{t=2}^n \ell_t(\boldsymbol{\theta})  \\
		\ell_t(\boldsymbol{\theta})  & =  y_t[\log\lambda_t-\log(\lambda_t+\phi_t)]+\phi_t[\log\phi_t-\log(\lambda_t+\phi_t)]
		+\log\Gamma(y_t+\phi_t)-\log\Gamma(\phi_t),
	\end{aligned}
\end{equation} 	
for $t=2,\ldots,n$, where we have omitted the dependence of $\lambda_t$ and $\phi_t$ on $\boldsymbol\theta$ for simplicity.

The computation of the log-likelihood function (\ref{loglikelihood}) depends on the initial values of $\lambda_0$ and $\phi_{0}$. If we denote by $\tilde{\lambda}_0$  and $\tilde{\phi}_{0}$ some fixed or random starting values of the process \eqref{DINGARCH}, then the observed log-likelihood function is given by  
	\begin{equation}
		\begin{aligned}
			\tilde{\ell}(\boldsymbol{\theta}) & = \sum_{t=1}^{n} \tilde{\ell}_{t}(\boldsymbol{\theta}) \\
			\tilde{\ell}_t(\boldsymbol{\theta})  & =  y_t[\log \tilde{\lambda}_t-\log( \tilde{\lambda}_t+ \tilde{\phi_t)}]+\tilde{\phi}_t[\log \tilde{\phi}_t-\log(\tilde{\lambda}_t+\tilde{\phi}_t)]
			+\log\Gamma(y_t+\tilde{\phi}_t)-\log\Gamma(\tilde{\phi}_t),
			\label{loglikelihood using the starting value}
		\end{aligned}
	\end{equation}
	with $\tilde{\lambda}_t=\beta_0+\beta_1 Y_{t-1}+\beta_2 \tilde{\lambda}_{t-1}$ and $\tilde{\phi}_t=\alpha_0+\alpha_1 Y_{t-1}+\alpha_2 \tilde{\phi}_{t-1}$.
	
The conditional maximum likelihood estimator (CMLE) of $\boldsymbol{\theta}$, say $\widehat{\boldsymbol{\theta}}$, is given by $\widehat{\boldsymbol{\theta}}=\mbox{argmax}_{\boldsymbol{\theta}\in\Theta}\tilde{\ell}(\boldsymbol{\theta})$, where $\Theta=(0,\infty)\times[0,\infty)^2\times(0,\infty)\times[0,\infty)^2$ denotes the parameter space.
	
In Lemma \ref{lemma:starting value} in the Appendix, we show that  the log-likelihood function  \eqref{loglikelihood using the starting value} is approximated by  \eqref{loglikelihood} with $(Y_t,\lambda_t, \phi_t)_{t \in \mathbb{Z}}$ the stationary and ergodic process obtained by Theorem \ref{Thm_StatErg};  
	see also \cite{francqzakoian2004}, \cite{StraumannandMikosh(2006)} and \cite{chrfok2014}, among others. Note that  \eqref{loglikelihood using the starting value} is the observed likelihood function but we will be employing  \eqref{loglikelihood} in what follows (including the proof of Theorem \ref{asymptotics} about the asymptotic normality of the CMLEs) for simplicity of notation and without loss of generality. In practice, we obtain initial values by employing the two first empirical moments of the data. This strategy has worked well; see Sec. \ref{Sec:Simulations} for empirical evidence.

The score function associated with the conditional log-likelihood function is 
\begin{eqnarray}
	{\bf U}(\boldsymbol{\theta})= \frac{\partial \ell(\boldsymbol{\theta})}{\partial\boldsymbol{\theta}}=\sum_{t=2}^n{\bf U}_t(\boldsymbol{\theta}),
	\label{eq:score function}
\end{eqnarray} 
with 
\begin{eqnarray*}
	{\bf U}_t(\boldsymbol{\theta})=\left(S_{1t}\dfrac{\partial\lambda_t}{\partial\beta_0},S_{1t}\dfrac{\partial\lambda_t}{\partial\beta_1},S_{1t}\dfrac{\partial\lambda_t}{\partial\beta_2},S_{2t}\dfrac{\partial\phi_t}{\partial\alpha_0},S_{2t}\dfrac{\partial\phi_t}{\partial\alpha_1},S_{2t}\dfrac{\partial\phi_t}{\partial\alpha_2}\right)^\top,
\end{eqnarray*}
and
\begin{eqnarray*}
	S_{1t}=\dfrac{\phi_t(y_t-\lambda_t)}{\lambda_t(\lambda_t+\phi_t)},\quad S_{2t}=-\dfrac{y_t-\lambda_t}{\lambda_t+\phi_t}+\log\left(\dfrac{\phi_t}{\lambda_t+\phi_t}\right)+\Psi(y_t+\phi_t)-\Psi(\phi_t),
\end{eqnarray*}
for $t=2,\ldots,n$, where $\Psi(x)=d\log\Gamma(x)/dx$, for $x>0$, is the digamma function. Explicit expressions for the derivatives involving $\lambda_t$ and $\phi_t$ are presented in the Supplementary Material. 

To establish the asymptotic normality of the CMLEs, the following lemma (its proof given in the Appendix) and proposition are necessary.

\begin{Lemma}\label{lemma:integrability}
For $\|{\bf A}\|_1+\|{\bf B}\|_1<1$, the $t$-th term of the score function ${\bf U}_t(\boldsymbol{\theta})$ is integrable for all $t\geq2$.
\end{Lemma}

\begin{Proposition}\label{martingale_diff}
 Let ${\bf U}_t(\boldsymbol{\theta})$ be the $t$-th term of the score function, for $t\geq2$. For $\|{\bf A}\|_1+\|{\bf B}\|_1<1$, $E({\bf U}_t(\boldsymbol{\theta}))=0$ for all $t\geq 2$, where the expectation is taken regarding the model with true parameter vector $\boldsymbol\theta$.
\end{Proposition}

\begin{proof} We will first show that $E({\bf U}_t(\boldsymbol{\theta})|\mathcal F_{t-1})=0$. This fact and the integrability of ${\bf U}_t(\boldsymbol{\theta})$ stated in Lemma \ref{lemma:integrability} will give us the desirable result.
 Since  $E(Y_t |\mathcal F_{t-1}) = \lambda_t$ by definition, we immediately obtain that $E(S_{1t}(\boldsymbol{\theta})|\mathcal F_{t-1})=0$. 
We  compute the conditional expectation of $S_{2t}$ given $\mathcal F_{t-1}$, which involves the term $E\left(\Psi(Y_t + \phi_t)|\mathcal F_{t-1}\right)$ and it is the most difficult to deal with. For $a>0$ and $|c|<1$, we have that 
\begin{eqnarray*}
\sum_{y=0}^\infty\dfrac{\Gamma'(y+a)}{y!}c^y=\dfrac{d}{da}\sum_{y=0}^\infty\dfrac{\Gamma(y+a)}{y!}c^y=\dfrac{d}{da}\dfrac{\Gamma(a)}{(1-c)^a}=\dfrac{\Gamma(a)}{(1-c)^a}\left\{\Psi(a)-\log(1-c)\right\},
\end{eqnarray*}
where $\Gamma'(a)=d\Gamma(a)/da$. Using the above result, it follows that
\begin{eqnarray*}
E\left(\Psi(Y_t + \phi_t)|\mathcal F_{t-1}\right)=\dfrac{(\lambda_t\phi_t^{-1}+1)^{-\phi_t}}{\Gamma(\phi_t)}\sum_{y=0}^\infty\dfrac{\Gamma'(y+\phi_t)}{y!}\left(\dfrac{\lambda_t}{\lambda_t+\phi_t}\right)^y=\Psi(\phi_t)-\log\left(\dfrac{\phi_t}{\lambda_t+\phi_t}\right).
\end{eqnarray*}
Thus 
\begin{eqnarray*}
E(S_{2t}(\boldsymbol{\theta})|\mathcal F_{t-1})&=&E\left(-\dfrac{Y_t-\lambda_t}{\lambda_t+\phi_t}+\log\left(\dfrac{\phi_t}{\lambda_t+\phi_t}\right)+\Psi(y_t+\phi_t)-\Psi(\phi_t)\Big|\mathcal F_{t-1}\right)\\
&=&-\dfrac{E(Y_t|\mathcal F_{t-1})-\lambda_t}{\lambda_t+\phi_t}+\log\left(\dfrac{\phi_t}{\lambda_t+\phi_t}\right)+E(\Psi(y_t+\phi_t)|\mathcal F_{t-1})-\Psi(\phi_t)=0.
\end{eqnarray*}
Hence, for all $t\geq2$, $E({\bf U}_t(\boldsymbol{\theta})|\mathcal F_{t-1})=0$.
From Lemma \ref{lemma:integrability} and under the condition $\|{\bf A}\|_1+\|{\bf B}\|_1<1$, ${\bf U}_t(\boldsymbol{\theta})$ is integrable. Therefore the conclusion holds true. 
\end{proof}

We next sketch  the proof for obtaining the asymptotic distribution of the CMLE for the NB tv-DINGARCH model. A more detailed technical result will be established in Theorem \ref{asymptotics}. Recall \eqref{eq:score function}. By using Proposition \ref{martingale_diff} and the square integrability of ${\bf U}(\boldsymbol\theta)$ (this is proved  using similar arguments to the proof of Proposition \ref{martingale_diff}), we can apply the Central Limit Theorem for martingale difference (see Corollary 3.1 from \cite{halhey1980}, for example) to obtain that 
\begin{eqnarray}\label{score_asympt}
n^{-1/2}{\bf U}(\boldsymbol\theta)\stackrel{d}{\longrightarrow} N({\bf0},\boldsymbol\Omega_1(\boldsymbol{\theta})),\quad \boldsymbol\Omega_1(\boldsymbol{\theta})\equiv\mbox{plim}_{n\rightarrow\infty}{\bf J}_1(\boldsymbol\theta),
\end{eqnarray}
as $n\rightarrow\infty$, 
\begin{eqnarray*}
	{\bf J}_1(\boldsymbol\theta)\equiv\dfrac{1}{n}\sum_{t=2}^n\mbox{Var}({\bf U}_t(\boldsymbol\theta)|\mathcal F_{t-1})=\dfrac{1}{n}\sum_{t=2}^n
	\begin{pmatrix}
		b_t\dfrac{\partial\lambda_t}{\partial {\boldsymbol\theta}_{1}}\dfrac{\partial\lambda_t}{\partial{\boldsymbol\theta}_{1}}^\top & 	{\bf 0} \\
		{\bf 0}  &  l_t\dfrac{\partial\phi_t}{\partial{\boldsymbol\theta}_{2}}\dfrac{\partial\phi_t}{\partial{\boldsymbol\theta}_{2}}^\top
	\end{pmatrix},
\end{eqnarray*}	
with  
\begin{eqnarray*}
 b_t=\dfrac{\phi_t}{\lambda_t(\lambda_t+\phi_t)},\quad l_t=E\left(\Psi^2(Y_t+\phi_t)|\mathcal F_{t-1}\right)-\dfrac{\lambda_t}{\phi_t(\lambda_t+\phi_t)}-\left(\psi(\phi_t)-\log\left(\dfrac{\phi_t}{\lambda_t+\phi_t}\right)\right)^2, 
\end{eqnarray*}
for $t=2,\ldots,n$, where we have used that $E\left(Y_t\Psi(Y_t+\phi_t)|\mathcal F_{t-1}\right)=\lambda_t\left\{\Psi(\phi_t+1)-\log\left(\dfrac{\phi_t}{\lambda_t+\phi_t}\right)\right\}$ and that $\Psi(x+1)=\Psi(x)+x^{-1}$, for $x>0$.

Then, we apply the Law of Large Numbers for stationary and ergodic sequences (see Thm. 6.28 of \cite{breiman1992}) to obtain that 
\begin{eqnarray}\label{hessian_asympt}
n^{-1}\dfrac{\partial{\bf U}(\boldsymbol\theta)}{\partial\boldsymbol\theta}\stackrel{p}{\longrightarrow} \boldsymbol\Omega_2(\boldsymbol{\theta}),\quad \boldsymbol\Omega_2(\boldsymbol{\theta})\equiv\mbox{plim}_{n\rightarrow\infty}{\bf J}_2(\boldsymbol\theta),
\end{eqnarray}
as $n\rightarrow\infty$, where
\begin{eqnarray*}
	{\bf J}_2(\boldsymbol\theta)\equiv\dfrac{1}{n}\sum_{t=2}^nE(-\nabla{\bf U}_t(\boldsymbol\theta)|\mathcal F_{t-1})=\dfrac{1}{n}\sum_{t=2}^n
	\begin{pmatrix}
		b_t\dfrac{\partial\lambda_t}{\partial {\boldsymbol\theta}_{1}}\dfrac{\partial\lambda_t}{\partial{\boldsymbol\theta}_{1}}^\top & 	{\bf 0} \\
		{\bf 0}  &  d_t\dfrac{\partial\phi_t}{\partial{\boldsymbol\theta}_{2}}\dfrac{\partial\phi_t}{\partial{\boldsymbol\theta}_{2}}^\top
	\end{pmatrix},	
\end{eqnarray*}	
with
\begin{eqnarray*}
 d_t=\Psi'(\phi_t)-E\left(\Psi'(Y_t + \phi_t)|\mathcal F_{t-1}\right)-\dfrac{\lambda_t}{\phi_t(\lambda_t+\phi_t)},\quad t=2,\ldots,n, 
\end{eqnarray*}
where $\Psi'(x)=d\Psi(x)/dx$ is the trigamma function. By using the facts that $\Psi'(x)=\dfrac{\Gamma''(x)}{\Gamma(x)}-\Psi^2(x)$ and 
\begin{eqnarray*}
	\sum_{y=0}^\infty\dfrac{\Gamma''(y+a)}{y!}c^y=\dfrac{\Gamma(a)}{(1-c)^a}\left\{\Psi'(a)+[\Psi(a)-\log(1-c)]^2\right\},
\end{eqnarray*}
for $a>0$ and $|c|<1$, we obtain that $l_t=d_t$ for all $t$. Therefore, ${\bf J}_1(\boldsymbol{\theta})={\bf J}_2(\boldsymbol{\theta})$ and $\boldsymbol\Omega_1(\boldsymbol{\theta})=\boldsymbol\Omega_2(\boldsymbol{\theta})$.

By combining the above results and using Taylor's expansion to obtain that 
\begin{eqnarray}\label{asympnorm}
\sqrt{n}(\widehat{\boldsymbol\theta}-\boldsymbol\theta)\stackrel{d}{\longrightarrow} N({\bf0},\boldsymbol\Sigma(\boldsymbol\theta))\,\,\mbox{as}\,\, n\rightarrow\infty,
\end{eqnarray}
with $\boldsymbol\Sigma(\boldsymbol\theta)=\boldsymbol\Omega_2^{-1}(\boldsymbol\theta)\boldsymbol\Omega_1(\boldsymbol\theta)\boldsymbol\Omega_2^{-1}(\boldsymbol\theta)=\boldsymbol\Omega_1^{-1}(\boldsymbol\theta)$.
We have that $\widehat{\boldsymbol\Sigma}={\bf J}_1^{-1}(\widehat{\boldsymbol\theta})$ is a consistent estimator for $\boldsymbol\Sigma$. 

There are required additional technical conditions to ensure the consistency and asymptotic normality of the CML estimators, which are provided in the next theorem, whose  proof can be found in the Appendix. We first  consider lower and upper bounds for the possible values of the components of the parameter vector by introducing (recall that $\theta_{3}=\beta_{2}$ and $\theta_{6}=\alpha_{2}$)

\begin{eqnarray}\label{B-set}
	\mathcal B=\Bigl\{\boldsymbol\theta:     
	0 < \theta_{i,low} \leq \theta_i \leq \theta_{i,up}, i,=1, \ldots,6, ~\mbox{and} ~ \theta_{3,up} < 1, \theta_{6,up} < 1\Bigr\}.
\end{eqnarray}

\begin{Theorem}\label{asymptotics}
	Let $\{Y_t\}_{t=1}^n$ follow the  linear negative binomial tv-DINGARCH(1,1,1,1) process. Assume that $\|{\bf A}\|_1+\|{\bf B}\|_1<1$ and that the true value of $\boldsymbol\theta$ is an interior point of $\mathcal B$. Then, there exists an open set $\mathcal A\subset\mathcal B$ such that the conditional log-likelihood function $\ell(\cdot)$ has a global maximum point on $\mathcal A$, say $\widehat{\boldsymbol\theta}$, with probability tending to 1 as $n\rightarrow\infty$. Furthermore, $\widehat{\boldsymbol\theta}$ is consistent for $\boldsymbol\theta$ and satisfies (\ref{asympnorm}). 
\end{Theorem}

\subsection{Monte Carlo Simulation}
\label{Sec:Simulations}

The finite-sample behavior of the NB $\mbox{tv-DINGARCH}(1,1,1,1)$ model CMLEs is investigated next. To that end, we conduct a Monte Carlo study with 1000 replications of trajectories with length $n=500, 1000$ and four parameter configurations. Two configurations are reported in what follows, and additional results are given in the Supplementary Material. The true values of $\boldsymbol{\theta}$ in the first two configurations (Settings I and II) are $\boldsymbol{\theta}=(\beta_0,\alpha_0,\beta_1,\beta_2,\alpha_1,\alpha_2)^\top=(15, 0.5, 0.2, 0.25, 0.1, 0.3)^\top$ and $\boldsymbol{\theta}=(\beta_0,\alpha_0,\beta_1,\beta_2,\alpha_1,\alpha_2)^\top=(3, 0.1, 0.3, 0.15, 0.2, 0.3)^\top$, respectively. Stationarity and uniqueness of the simulated tv-DINGARCH processes are guaranteed by ensuring that Theorem \ref{Thm_StatErg} holds when  $p=1$, that is, $\|{\bf A}\|_1 + \|{\bf B}\|_1=\beta_2+\alpha_2+\beta_1+\alpha_1< 1$. 
All results are obtained by employing restricted optimization such that stationarity conditions are satisfied. 

Table \ref{tab:simulated} summarises the results obtained for settings I and II. Empirical means and standard deviation (SD) of the CMLEs are reported by sample size. The results show adequate performance of the estimation method, with empirical Monte Carlo means close to the values and standard deviation decreasing with the increase in sample size, as expected.

\begin{table}[h!]
	\centering
	\small
	\begin{tabular}{@{}lllcccccc@{}}
		\toprule
		\multicolumn{3}{l}{Setting I}                                  & \textbf{$\beta_0=15$} & \textbf{$\alpha_0=0.5$} & \textbf{$\beta_1=0.2$} & \textbf{$\beta_2=0.25$} & \textbf{$\alpha_1=0.1$} & \textbf{$\alpha_2=0.3$} \\ \midrule
		\multirow{6}{*}{}  & \multirow{3}{*}{$n = 500$} & Mean & 15.743 & 0.701 & 0.196 & 0.226 & 0.102 & 0.254   \\
		&                          & SD   & 3.900 & 0.581 & 0.045 & 0.153 & 0.023 & 0.153     \\ \vspace{0.1cm} \\

		& \multirow{3}{*}{$n=1000$}   & Mean & 15.491 & 0.590 & 0.199 & 0.233 & 0.100 & 0.285    \\
		&                          & SD   &2.973 & 0.408 & 0.031 & 0.115 & 0.018 & 0.121   \\
 \midrule
				\multicolumn{3}{l}{Setting II}                                  & \textbf{$\beta_0=3$} & \textbf{$\alpha_0=0.1$} & \textbf{$\beta_1=0.3$} & \textbf{$\beta_2=0.15$} & \textbf{$\alpha_1=0.2$} & \textbf{$\alpha_2=0.3$} \\ \midrule
		\multirow{6}{*}{} & \multirow{3}{*}{$n = 500$} & Mean & 3.126 & 0.118 & 0.298 & 0.130 & 0.207 & 0.284 \\
		&                          & SD   &0.515 & 0.052 & 0.045 & 0.092 & 0.036 & 0.081 \\ \vspace{0.1cm} \\
		& \multirow{3}{*}{$n=1000$}   & Mean & 3.098 & 0.107 & 0.299 & 0.135 & 0.203 & 0.293        \\
		&                          & SD   & 0.410 & 0.034 & 0.031 & 0.073 & 0.024 & 0.057       \\ \bottomrule

	\end{tabular}
\caption{Empirical mean and standard deviation (SD) of the Monte Carlo estimates for the NB tv-DINGARCH model parameters under Settings I and II. Results are based on 1000 Monte Carlo replications.}\label{tab:simulated}
\end{table}

Figure \ref{asymptotics_set2} illustrates the asymptotic normality of the CMLEs for setting I. Non-parametric density estimator plots of the standardized parameter estimates are displayed alongside the standard Gaussian density (solid line). Dashed and dotted curves indicate densities estimated from the experiments carried out with $n=500$ and $n=1000$, respectively. The density curves of the parameter estimates are mostly overlapping, but some improvement can be seen with the increase in the sample size. Additional density plots for the other parameter configurations can be found in the Supplementary Material.

\begin{figure}
\centering
\includegraphics[width=0.75\linewidth]{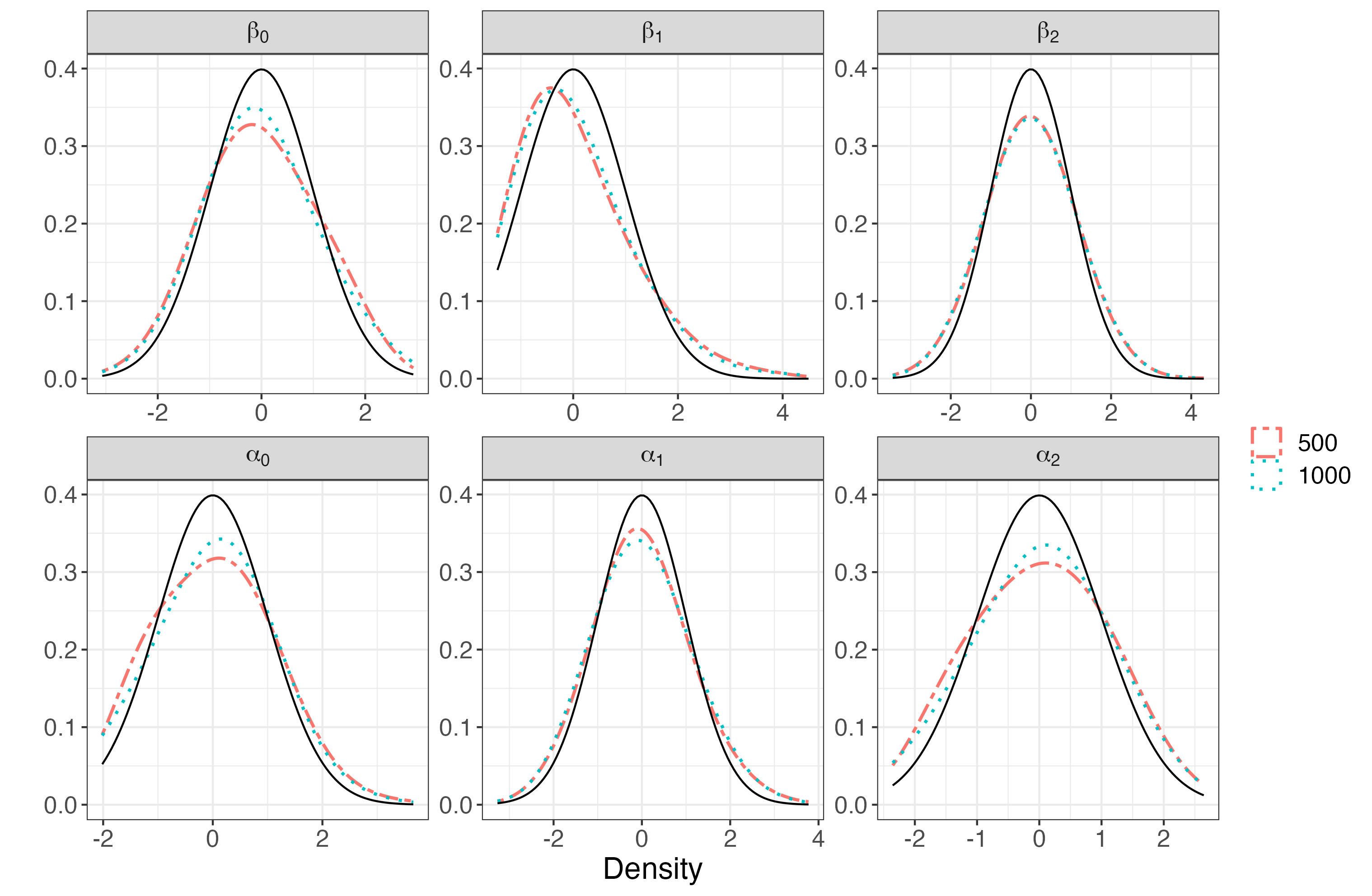}
\caption{Non-parametric density plots of standardized parameter estimates due to Setting I when  $n=500$ and $n = 1000$. The solid line corresponds to the standard Gaussian density function. }\label{asymptotics_set2}
\end{figure}

\section{Testing Constant Dispersion}\label{sec:test}

Testing constant value of the dispersion parameter, i.e. $\phi_t=\alpha_0$ for all $t$ under model \eqref{DINGARCH}, is equivalent to testing $H_0: \alpha_1 = \alpha_2 = 0$ against $H_1: \alpha_1 > 0\,\,\mbox{or}\,\,  \alpha_2 >  0$. Note that the null hypothesis implies that the parameters $\alpha_{1}$ and $\alpha_{2}$ belong to the boundary of the parameter space. Hence, this  is a non-standard testing problem; see, for instance, \cite{sellia1987}, \cite{and2001}, and \cite{crarup2004}. 

We consider the likelihood ratio test  to develop and compare heuristically two parametric bootstrap methods; the classical or unrestricted, and the restricted bootstrap recently developed by \cite{cavaliere2016}. The first method considers the usual parametric bootstrap \citep{efrtib1993} replications based on the unrestricted CMLEs, while the latter uses the CMLEs under $H_0$.  Algorithm \ref{Alg:bootstrap}  specifies computation  of the likelihood ratio  test's p-value with $B$ replications of the restricted or unrestricted bootstrap.

\begin{algorithm}
	\DontPrintSemicolon
	\SetAlgoLined
	\SetNoFillComment
	\LinesNotNumbered 
	\KwIn{$\boldsymbol{Y}$ observed count time series data \\
		$B$ bootstrap replications \\
		$\alpha$ significance level }
	
	\vspace{0.2cm}
	
	1. Obtain $\widehat{\boldsymbol{\theta}}_{\mathcal{H}_0}$ and $\boldsymbol{\widehat{\theta}}$, the model CMLEs under $\mathcal{H}_0$ and $\mathcal{H}_1$;
	
	2. Compute the observed likelihood ratio $LR = -2(\ell(\widehat{\boldsymbol{\theta}}_{\mathcal{H}_0}) - \ell(\boldsymbol{\widehat{\theta}}))$;

	\For{$b \leftarrow 1:B$}{
		
		\vspace{0.2cm}
		
		3A) $\boldsymbol{Y}_b \sim \mbox{NB tv-DINGARCH}(\widehat{\boldsymbol{\theta}}_{\mathcal{H}_0})$; \tcp*[l]{ if restricted bootstrap} 
		
		3B) $\boldsymbol{Y}_b \sim \mbox{NB tv-DINGARCH}(\widehat{\boldsymbol{\theta}})$;\tcp*[l]{if unrestricted bootstrap}
		
		\vspace{0.2cm}
		
		4. Obtain  $\widehat{\boldsymbol{\theta}}_{\mathcal{H}_0}^b$ and $\boldsymbol{\widehat{\theta}}^b$ fitting NB tv-DINGARCH models to $\boldsymbol{Y}_b$ under the null and alternative hypothesis;
		
		5. Let $LR^b = -2(\ell(\widehat{\boldsymbol{\theta}}^b_{\mathcal{H}_0}) - \ell(\boldsymbol{\widehat{\theta}}^b))$, the replicated LR statistic;
		
	}
	
	6. If $p_B = \sum_{b=1}^B I\{ LR^b > LR\}/B < \alpha$ reject $\mathcal{H}_0$;
	
	\caption{Bootstrap likelihood ratio test of constant ($\mathcal{H}_0: \alpha_1 = \alpha_2 = 0$) versus time-varying dispersion ($\mathcal{H}_1: \alpha_1 >  0 \mbox{ or } \alpha_2 >  0$) for a NB tv-DINGARCH model. Alternatives to step 3 yield restricted (3A) or unrestricted bootstrap (3B) estimators of the test's p-value, $p_B$, where $B$ is the number of replications.}\label{Alg:bootstrap}
\end{algorithm}

We use a Monte Carlo simulation study to investigate how these methods achieve the desirable significance levels. Time series from the NB tv-DINGARCH process with 200  observations are simulated under the null hypothesis using four different (varying) settings of the parameter vector $(\beta_0,\beta_1,\beta_2,\alpha_0)^\top$ as follows: (\textbf{C1}) $\beta_0 =2, \alpha_0 = 1$, $\beta_1 = 0.4, \beta_2 = 0.3$, (\textbf{C2}) $\beta_0 =2, \alpha_0 = 1$, $\beta_1 = 0, \beta_2 = 0$, (\textbf{C3}) $\beta_0 =3, \alpha_0 = 0.5$, $\beta_1 = 0.3, \beta_2 = 0.4$, and (\textbf{C4}) $\beta_0 =3, \alpha_0 = 0.5$, $\beta_1 = 0, \beta_2 = 0$. For each configuration, 500 Monte Carlo replications are used to calculate the empirical significance levels by employing the competing methodologies. The number of replications used to estimate bootstrap $p$-values is $B=500$. Table \ref{boot_sig_level} displays the proportion of times that the restricted and unrestricted bootstrap procedures rejected the null hypothesis. Parameter configurations are set in a way that \textbf{C2} and \textbf{C4} are variations of \textbf{C1} and \textbf{C3} that do not include effects on the mean. Additional evidence is provided in the Supplementary Material.

\begin{table}[ht!]
	\centering
	\small
	\begin{tabular}{@{}cccc@{}}
		\toprule
		\textbf{Configuration}                                                                                                      & \textbf{\begin{tabular}[c]{@{}c@{}}Nominal \\ level\end{tabular}} & \textbf{\begin{tabular}[c]{@{}c@{}}Restricted\\ Bootstrap\end{tabular}} & \textbf{\begin{tabular}[c]{@{}c@{}}Unrestricted\\ Bootstrap\end{tabular}} \\ \midrule
		\multirow{2}{*}{\begin{tabular}[c]{@{}c@{}} \textbf{C1}: $\beta_0 =2, \alpha_0 = 1$\\ $\beta_1 = 0.4, \beta_2 = 0.3$\end{tabular}}   & 0.05                                                                  & 0.046                                                                   & 0.000                                                                     \\
		& 0.10                                                                  & 0.088                                                                   & 0.002                                                                     \\ \midrule
		\multirow{2}{*}{\begin{tabular}[c]{@{}c@{}} \textbf{C2}: $\beta_0 =2, \alpha_0 = 1$\\ $\beta_1 = 0, \beta_2 = 0$\end{tabular}}        & 0.05                                                                  & 0.052                                                                   & 0.012                                                                     \\
		& 0.10                                                                  & 0.098                                                                   & 0.036                                                                     \\  \midrule
		\multirow{2}{*}{\begin{tabular}[c]{@{}c@{}} \textbf{C3}: $\beta_0 =3, \alpha_0 = 0.5$\\ $\beta_1 = 0.3, \beta_2 = 0.4$\end{tabular}} & 0.05                                                                  & 0.064                                                                   & 0.002                                                                     \\
		& 0.10                                                                  & 0.104                                                                   & 0.002                                                                     \\  \midrule
		\multirow{2}{*}{\begin{tabular}[c]{@{}c@{}} \textbf{C4}: $\beta_0 =3, \alpha_0 = 0.5$\\ $\beta_1 = 0, \beta_2 = 0$\end{tabular}}     & 0.05                                                                  & 0.042                                                                   & 0.022                                                                     \\
		& 0.10                                                                  & 0.094                                                                   & 0.058                                                                     \\ \bottomrule
	\end{tabular}\caption{Achieved  significance levels of likelihood ratio test  employing Algorithm 1 for various parameter configurations.}\label{boot_sig_level} 
\end{table}

Notably, the restricted parametric bootstrap achieves levels close to nominal significance levels, something that occurs under all four parameter configurations investigated in this study. The unrestricted parametric bootstrap does not provide satisfactory results and it underestimates the significance levels. Hence, in applications, it is important to employ a restricted bootstrap for testing constant dispersion in the tv-DINGARCH models.

The results obtained from both restricted and unrestricted bootstrap were already expected. As discussed by \cite{cavaliere2016}, it is well-known that the unrestricted (classical) bootstrap does not apply when one parameter belongs to the boundary of the parameter space. This is indeed the case where the null hypothesis is that $H_0: \alpha_1=\alpha_2=0$, both parameters belong to the boundary space. The restricted bootstrap proposed by  \cite{cavaliere2016} overcomes this issue and provides valid results as theoretically proven in their paper and empirically illustrated here for DINGARCH models.

\section{Measles Data Analysis}\label{sec:application}

The linear negative binomial tv-DINGARCH model is now applied to modeling the weekly number of reported measles infections in North Rhine-Westphalia, Germany. The series is observed between January 2001 and May 2013 (646 observations). These data are publicly available in the \texttt{R} package \texttt{tscount}.  Figure \ref{fig:measles_plot} displays $\{Y_t \}_{t=1}^{646}$ on the left, and the series autocorrelation function on the right.

\begin{figure}[ht!]
\centering
\includegraphics[width = 1\linewidth]{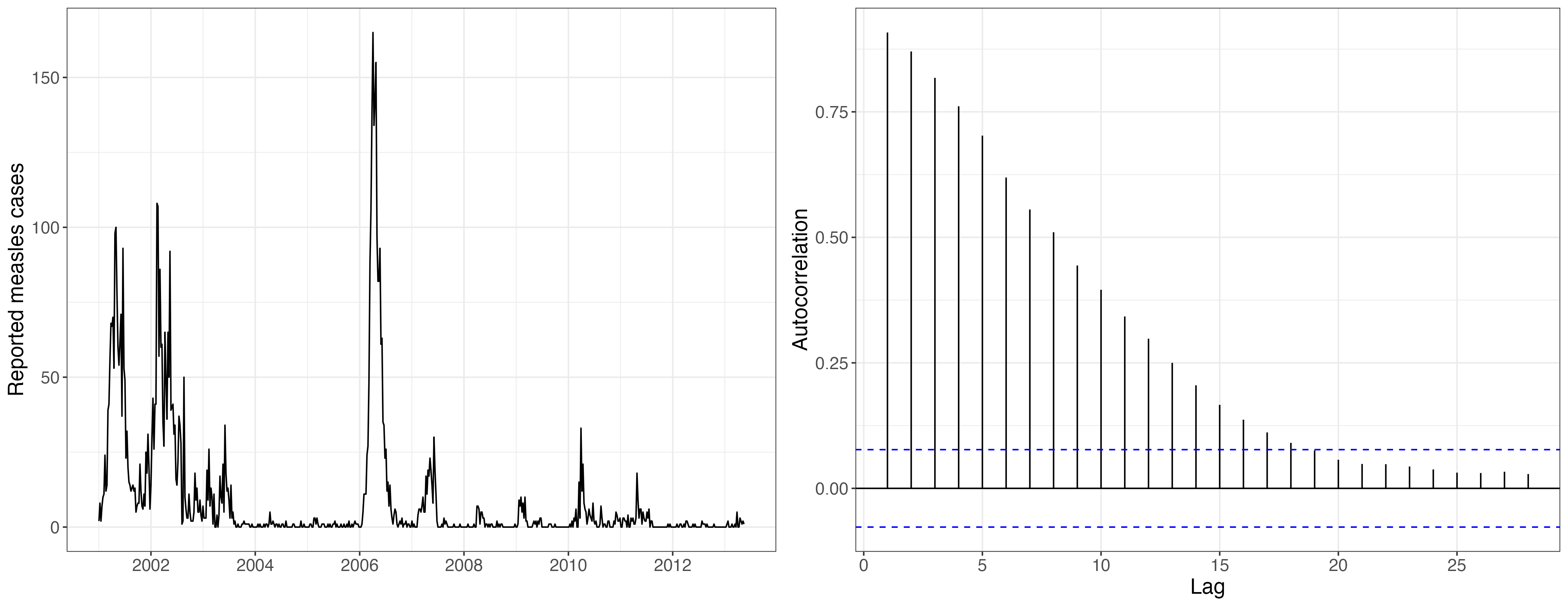}
\caption{On the left, weekly cases of measles reported in North Rhine-Westphalia, Germany, between January 2001 and May 2013. The autocorrelation function of the series is shown on the right. }\label{fig:measles_plot}
\end{figure}

The fitting of the NB tv-DINGARCH model will be compared to the fit of the ordinary NB INGARCH model with an identical mean time series structure. The purpose of this exercise is to assess how a time-varying dispersion changes the model adequacy to this real data. To this end, we will evaluate goodness-of-fit and predictions for each model, i.e. $\mathcal{M}_{tv}$ will denote the NB tv-DINGARCH model while $\mathcal{M}_{ord}$ denotes the ordinary case (which assumes a constant dispersion).

Results for $\mathcal{M}_{tv}$ and $\mathcal{M}_{ord}$ models are summarized by Table \ref{tab:cmles}. Conditional maximum likelihood estimates of the parameters, their standard errors (SE) in parenthesis, and associated approximate 95\% confidence intervals are provided. Uncertainty quantification relies on 500 replications of parametric bootstrap for both $\mathcal{M}_{tv}$ and $\mathcal{M}_{ord}$.  Table \ref{tab:cmles} shows that $\mathcal{M}_{ord}$ and $\mathcal{M}_{tv}$ are in close  agreement concerning the INGARCH mean structure as $\beta_0$, $\beta_1$ and $\beta_2$ assume similar values. To compare between models, we compute model information criteria and perform the constant dispersion test developed in Section \ref{sec:test}. The values of AIC and BIC are respectively $2670.568$ and $2697.393$ for $\mathcal{M}_{tv}$ and $2797.216$ and $2815.099$ for $\mathcal{M}_{ord}$, both supporting time-varying dispersion.  The likelihood ratio test for $H_0: \alpha_1 = \alpha_2 = 0$ versus $H_1$: \textit{$\alpha_1\neq0$ or $\alpha_2\neq0$} is in line with the model choice conclusion, returning high evidence ($\mbox{p-value} < 10^{-5}$) in favor of $\mathcal{M}_{tv}$ over $\mathcal{M}_{ord}$.

\begin{table}
\centering
\begin{tabular}{ccccc} \hline 
\multirow{2}{*}{Parameter} & \multicolumn{2}{c}{tv-DINGARCH}     & \multicolumn{2}{c}{INGARCH}             \\
                           & Estimate (se)    & 95\% CI          & Estimate (se)    & 95\% CI              \\  \hline 
$\beta_0$                  & 0.259 (0.048) & (0.164, 0.353) & 0.194 (0.127)   & (0.000, 0.444) \\
$\beta_1$                  & 0.579 (0.034)    & (0.512, 0.645)  &  0.583 (0.090)    & (0.407, 0.759)    \\
$\beta_2$                  & 0.342 (0.039)    & (0.265, 0.419)   &  0.390 (0.094)   & (0.205, 0.574)    \\
\hline
$\alpha_0$           & 0.775 (0.057)          & (0.663, 0.886)  & 0.736 (0.110) & (0.521, 0.952)    \\
$\alpha_1$           & 0.079 (0.007)         & (0.065, 0.093) & -- & --                                         \\
$\alpha_2$           & 0.000 (0.010)         & (0.000, 0.020)  & -- & --                                                           \\ 
\hline
\end{tabular}
	\caption{Conditional maximum likelihood estimates, standard errors (se), and 95\% confidence intervals for $\mathcal{M}_{tv}$ and $\mathcal{M}_{ord}$ models applied to weekly counts of measles in Germany.}\label{tab:cmles}
\end{table}

In addition, we consider Probability Integral Transform (PIT) plots \citep{czaetal2009}, an approach that enables the comparison of count time series  models via their predictive distributions. A model providing a good fit to the data in this aspect will render a PIT plot resembling a uniform distribution, where major deviations typically indicate problems of overdispersion or underdispersion of the model's predictive distribution. These graphs are shown in Figure \ref{pit_models_noint_IE} for both models and while the plot  for  $\mathcal{M}_{tv}$ resembles uniformity, as desired, the upside-down U-shaped PIT of $\mathcal{M}_{ord}$ indicates a non adequate fit.

\begin{figure}[ht!]
\centering
\includegraphics[width = 1\linewidth]{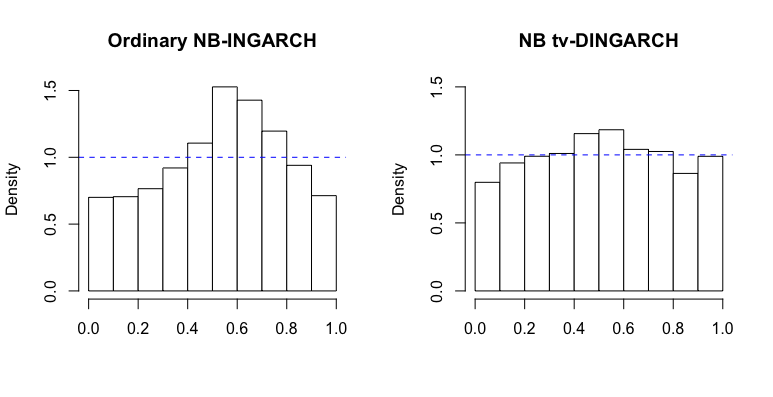}
\caption{PIT plots for $\mathcal{M}_{ord}$ (left) and $\mathcal{M}_{tv}$ (right).}\label{pit_models_noint_IE}
\end{figure}

It is often of interest to practitioners to select between models according to their forecasting power. To that end, we consider a forecasting exercise that explores one-step-ahead (OSA) prediction from $\mathcal{M}_{tv}$ and $\mathcal{M}_{ord}$. This is done in a recursive manner to provide multiple OSA predictions from each model. We start by defining the initial training data from the beginning of the study until October 2004 (week 4) which contains 200 observations. Both models are fitted to the training set and used to predict the next week's counts (week 1 of November 2004) via the conditional median and mode of the distributions. The conditional (or predictive) distributions are $\mbox{NB}(\hat{\mu}_{t+1}, \hat{\phi}_{t+1})$ and $\mbox{NB}(\hat{\mu}_{t+1}, \hat{\phi})$ respectively for $\mathcal{M}_{tv}$ and $\mathcal{M}_{ord}$ and $t =200$ at this step. Once the prediction is obtained, we add week 1 of November 2004 to the training set, refit both models, and gather the new OSA predictions. Proceeding until the end of the study period gives a total of 446 predictions from each model. A pseudocode describing the steps to this prediction exercise is given in Algorithm \ref{Alg:OSA}.

\begin{algorithm}[ht!]
	\DontPrintSemicolon
	\SetAlgoLined
	\SetNoFillComment
	\LinesNotNumbered 
	\KwIn{$\boldsymbol{Y}$ observed count trajectory of length $n$ \\
		
		$n_0$ starting point of prediction exercise}
	
	\vspace{0.2cm}
	
	0. $s \leftarrow n_0$;
	
	\While{$s < n$}{
		
		\vspace{0.5cm}
		
		1. $\boldsymbol{Y}(s) \leftarrow Y_{[1:s]}$; \tcp*[l]{train data}
		
		2. Fit the NB tv-DINGARCH model to $\boldsymbol{Y}(s)$;
		
		3. From 2, put together the CMLEs $\widehat{\boldsymbol{\theta}}(s)$ and the fitted  $\widehat{\boldsymbol{\lambda}}(s)$ and  $\widehat{\boldsymbol{\phi}}(s)$ of step $s$;
		
		\vspace{0.5cm}
		
		4. Obtain the OSA mean $\widehat{\mu}_{s+1} = \widehat{\beta}_0^s + \widehat{\beta}_1^s Y_{s}  +   \widehat{\beta}_2^s \widehat{\mu}_s$ and $\widehat{\phi}_{s+1} = \widehat{\alpha}_0^s + \widehat{\alpha}_1^s Y_{s} +   \widehat{\alpha}_2^s \widehat{\phi}_s$;
		
		\vspace{0.5cm}
		
		5A) Compute the mode of the $\mbox{NB}( \widehat{\mu}_{s+1},\widehat{\phi}_{s+1})$ distribution and denote it by  $\widehat{Y}_{s+1}$;\tcp*[l]{prediction via mode}
		
		5B) Compute the median of the $\mbox{NB}( \widehat{\mu}_{s+1},\widehat{\phi}_{s+1})$ distribution and denote it by  $\widehat{Y}_{s+1}$;\tcp*[l]{prediction via median}		
		\vspace{0.5cm}		
		$s = s+1$;  \tcp*[l]{increment step}
	}
	\textbf{Output:} $\boldsymbol{\widehat{Y}}$ vector of $(n-n_0)$ OSA predictions.
\caption{Recursive algorithm for obtaining $\boldsymbol{\widehat{Y}}$, the one-step-ahead (OSA) predicted values of $\boldsymbol{Y}_{[n_0+1:n]}$. The training data at iteration $s$, $\boldsymbol{Y}(s)$, is incremented and the model is refitted to obtain the OSA forecast $\widehat{Y}_{s+1}$. Steps 5A) and 5B) provide alternative prediction methods via the mode and median.}\label{Alg:OSA}
\end{algorithm}

Algorithm \ref{Alg:OSA} is presented in terms of the $\mathcal{M}_{tv}$ model but applies  similarly for the case of $\mathcal{M}_{ord}$. In this case, $\mathcal{M}_{ord}$ is fitted in step two, and the predictions in 5A) and 5B) take in the (fixed) dispersion parameter estimated in step 2. By employing the Algorithm with $\mathcal{M}_{tv}$ and $\mathcal{M}_{ord}$, their predictions are summarized via the root mean square forecasting error (RMSFE). Let $n_0$ denote the time point chosen to start the prediction exercise; in this case $n_0=200$. The RMSFE of the forecasting step $t$ is 
$$
\mbox{RMSFE}_{t} = \sqrt{\frac{1}{t-n_0}\sum_{s=n_0+1}^{t}(Y_s - \widehat{Y}_s)^2},
$$ 
where $Y_s$ and $\widehat{Y}_s$,  
are the observed and predicted counts, respectively.  Calculation  of RMSFE, from the 446 total predictions from $\mathcal{M}_{tv}$ and $\mathcal{M}_{ord}$, yields   Figure \ref{MSFE_figure}. On the left, $\widehat{Y}_s$ is the mode of the conditional distribution and on the right, it is taken as the median. Prediction by the median yields a smaller prediction error in comparison to the mode, and in both cases the tv-DINGARCH model produces the smallest RMSFE values. 
\begin{figure}[ht!]
	\centering
	\includegraphics[width = 1\linewidth]{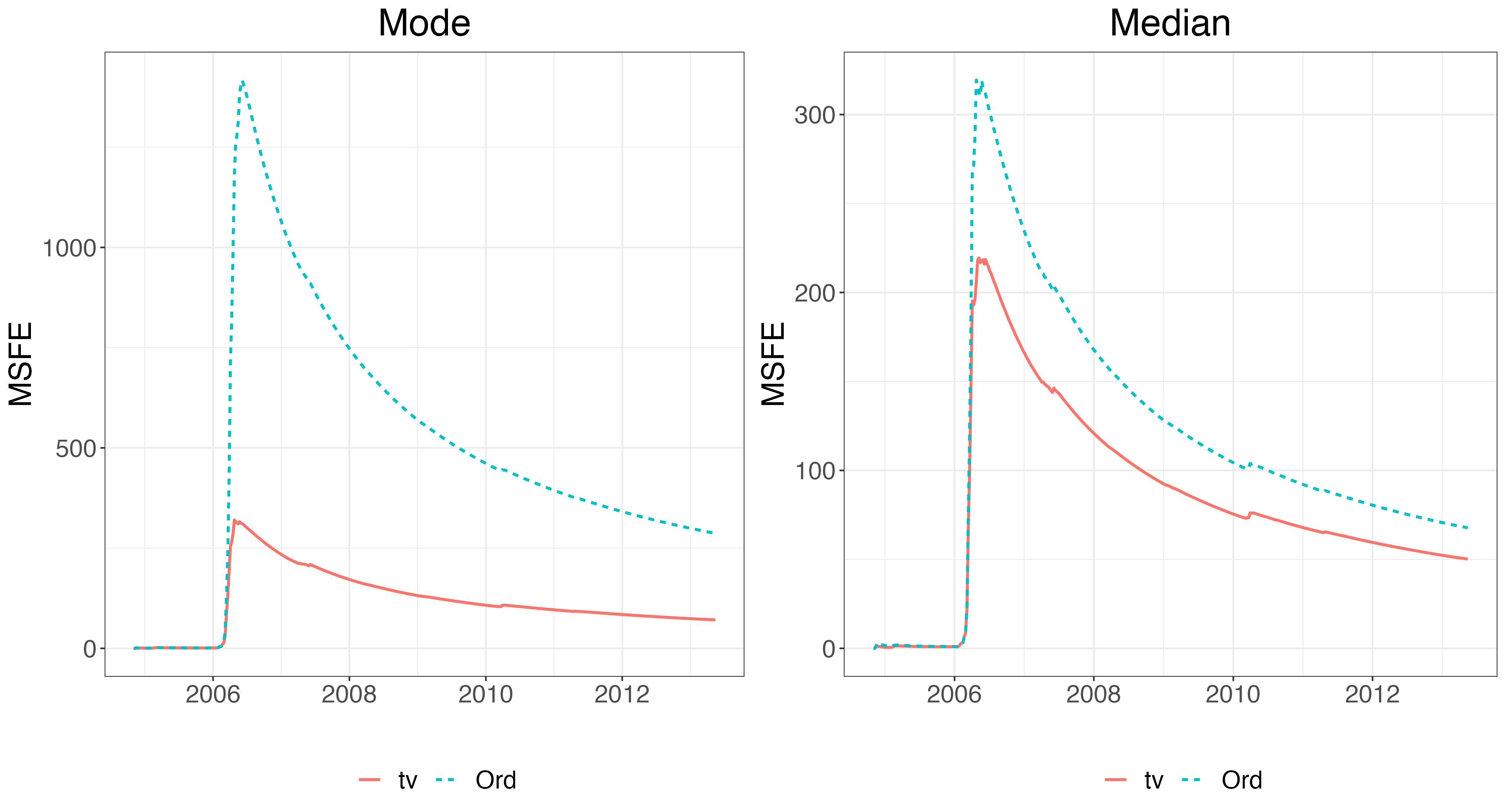}
	\caption{RMSFE of predictions obtained with the fit of ordinary (dashed lines) and time-varying dispersion (solid lines) INGARCH models to the weekly count of measles in North Rhine-Westphalia, Germany. On the left, predicted values are the mode of the predictive distributions, whereas the median is taken on the right.}\label{MSFE_figure}
\end{figure}

The analysis  of the weekly number of measles cases in North Rhine-Westphalia, Germany, demonstrated that the time-varying dispersion is a promising and important extension of the ordinary INGARCH processes that achieved improvement in terms of goodness-of-fit and forecasting.

\section{Conclusions}

We proposed a class of time-varying dispersion INGARCH (tv-DINGARCH) models and explored stochastic properties such as stationarity and ergodicity. Estimation of parameters was addressed through conditional maximum likelihood estimation (CMLE) and its associated asymptotic theory was established. Monte Carlo simulations were conducted to evaluate the performance of the CMLE. Moreover, we developed bootstrap methodologies to test for constant dispersion and showed via simulated studies that the restricted bootstrap is preferred over the unrestricted parametric one. We analyzed the weekly number of reported measles cases in North Rhine-Westphalia, Germany, from January 2001 to May 2013, and found that the tv-DINGARCH approach delivers much better results regarding goodness-of-fit and prediction when compared to the ordinary INGARCH model.

A log-linear version of our model allowing for the inclusion of covariates deserves future research. A first-order log-linear tv-DINGARCH process $\{Y_t\}_{t\geq1}$ allowing for covariates/exogenous time series is given as in Definition \ref{def:dingarch} with $\lambda_t\equiv\exp(\mu_t)$ and $\phi_t\equiv\exp(\nu_t)$, where
	\begin{eqnarray*}\label{struct_loglinear}
		\mu_t=\beta_0+\beta_1\log(Y_{t-1}+1)+\beta_2\mu_{t-1}+\boldsymbol\delta^\top{\bf X}_t,\quad	
		\nu_t=\alpha_0+\alpha_1\log(Y_{t-1}+1)+\alpha_2\nu_{t-1}+\boldsymbol\gamma^\top{\bf W}_t,
	\end{eqnarray*}	
$\beta_i,\alpha_i\in\mathbb R$ for $i=0,1,2$, with ${\bf X}_t$ and ${\bf W}_t$ covariates with associated real-valued coefficients $\boldsymbol\delta$ and $\boldsymbol\gamma$.  The ordinary log-linear INGARCH model \citep{foktjo2011} is obtained as a particular case from the log-linear tv-DINGARCH process by taking $\alpha_1=\alpha_2=0$ and $\boldsymbol\gamma=0$. This topic will be investigated in a future communication. Finally, it is worth mentioning that other mixed Poisson distributions rather than negative binomial can be used for our tv-DINGARCH formulation such as Poisson-inverse Gaussian distribution.

\renewcommand{\theequation}{A-\arabic{equation}}
\renewcommand{\theLemma}{A-\arabic{Lemma}}
\setcounter{equation}{0}
\setcounter{Lemma}{0}
\renewcommand{\thesubsection}{A-\arabic{subsection}}
\setcounter{subsection}{0}
\section*{Appendix}

\subsection{Proof of Theorem \ref{Thm_StatErg}}

 The key ingredient to establish the desired result is to prove that $\{\boldsymbol\xi_t\}_{t\in\mathbb Z}$ is an e-chain, that is, for any continuous function $w$ with compact support on $(0,\infty)^2$ and for every $\epsilon>0$, there exists $\eta>0$ such that, for ${\bf x},{\bf z}\in(0,\infty)^2$, $\|{\bf x}-{\bf z}\|<\eta$ implies that $\left|E(w(\boldsymbol\xi_k)|\boldsymbol\xi_0={\bf x})-E(w(\boldsymbol\xi_k)|\boldsymbol\xi_0={\bf z})\right|<\epsilon$ $\forall\, k\geq1$, where $\|\cdot\|$ is some norm.
 
We follow the arguments of  \cite{liu2012} and show explicitly the steps required for the proof based on model  \eqref{DINGARCH}; further details can be proved along the lines of  \cite{liu2012} and are omitted.

By working in the same manner as in \citet[pp.106]{liu2012} , we can prove that $\{\boldsymbol\xi_t\}_{t\in\mathbb Z}$ is a weak Feller chain and bounded in probability. Therefore, $\{\boldsymbol\xi_t\}_{t\in\mathbb Z}$ has at least one stationary distribution.
 
 Let $w$ be a continuous function with compact support on $(0,\infty)^2$ and assume, without loss of generality,   that it is   
 bounded by 1, i.e.  $|w|<1$. Consider $k=1$, ${\bf x}=(x_1,x_2),{\bf z}=(z_1,z_2)\in(0,\infty)^2$, and $\epsilon>0$ arbitrary. Denote by $f^{ mp}_{\boldsymbol\xi;\mathcal D}(\cdot)$ the probability function of a mixed Poisson distribution with mean $\lambda$, variance $\lambda+\lambda^2/\phi$, and mixing distribution $\mathcal D$, where $\boldsymbol\xi=(\lambda,\phi)^\top$. It follows that	
 \begin{eqnarray}
 	&&		\left|E(w(\boldsymbol\xi_1)|\boldsymbol\xi_0={\bf x})-E(w(\boldsymbol\xi_1)|\boldsymbol\xi_0={\bf z})\right|\leq\nonumber\\
 	&&		\sum_{y=0}^\infty|w(\boldsymbol\tau+{\bf B}(y,y)^\top+{\bf A}\boldsymbol{\bf x})f^{mp}_{\bf x; \mathcal{D}}(y)-w(\boldsymbol\tau+{\bf B}(y,y)^\top+{\bf A}\boldsymbol{\bf z})f^{mp}_{\bf z; \mathcal{D}}(y)|\leq\nonumber\\
 	&&		\sum_{y=0}^\infty f^{mp}_{\bf x; \mathcal{D}}(y)|w(\boldsymbol\tau+{\bf B}(y,y)^\top+{\bf A}\boldsymbol{\bf x})-w(\boldsymbol\tau+{\bf B}(y,y)^\top+{\bf A}\boldsymbol{\bf z})|+\label{term_I}\\
 	&&	\sum_{y=0}^\infty |w(\boldsymbol\tau+{\bf B}(y,y)^\top+{\bf A}\boldsymbol{\bf z})|\,|f^{mp}_{\bf x; \mathcal{D}}(y)-f^{mp}_{\bf z; \mathcal{D}}(y)|.\label{term_II}
 \end{eqnarray} 		
 
 We now find an upper bound for the term (\ref{term_II}). Denote by $f^{pois}_{\lambda}(\cdot)$ the probability function of a Poisson distribution with mean $\lambda$. From the mixed Poisson stochastic representation, we have that $f^{mp}_{\bf x; \mathcal{D}}(y)=E(f^{pois}_{x_1Z_1}(y))$, where $Z_1$ is the associated latent random variables with distribution depending on $x_2$. Similar representation holds for $f^{mp}_{\bf z; \mathcal{D}}(y)$ in terms of an associated latent factor $Z_2$ with distribution depending on $z_2$. Using this and recalling the fact that $|\omega|<1$, we obtain that	
 \begin{eqnarray}\label{term_II_aux}	
 	\sum_{y=0}^\infty |w(\boldsymbol\tau+{\bf B}(y,y)^\top+{\bf A}\boldsymbol{\bf z})|\,|f^{mp}_{\bf x; \mathcal{D}}(y)-f^{mp}_{\bf z; \mathcal{D}}(y)|\leq \sum_{y=0}^\infty |f^{mp}_{\bf x}(y)-f^{mp}_{\bf z}(y)|=\nonumber\\
 	\sum_{y=0}^\infty |E(f^{pois}_{x_1Z_1}(y))-E(f^{pois}_{z_1Z_2}(y))|\leq E\left(\sum_{y=0}^\infty|f^{pois}_{x_1Z_1}(y)-f^{pois}_{z_1Z_2}(y)|\right).
 \end{eqnarray}
 
By using inequality (A.1) from \cite{wanetal2014} (see also \cite{liu2012}), we have that $\sum_{y=0}^\infty|f^{pois}_{x_1Z_1}(y)-f^{pois}_{z_1Z_2}(y)|\leq 2\left(1-\exp\{-|x_1Z_1-z_1Z_2|\}\right)$. Hence, (\ref{term_II_aux}) is bounded above by 
\begin{eqnarray*}
	&& E\left(\sum_{y=0}^\infty|f^{pois}_{x_1Z_1}(y)-f^{pois}_{z_1Z_2}(y)|\right)\leq 2E\left(1-e^{-|x_1Z_1-z_1Z_2|}\right)\leq\\
	&& 2E\left(1-e^{-|x_1-z_1|(Z_1+Z_2)}\right)\leq  2E\left(|x_1-z_1|(Z_1+Z_2)\right)=2|x_1-z_1|\leq 2\|{\bf x}-{\bf z}\|_p,
\end{eqnarray*}
for $p\in[1,\infty]$, where we have used in the second inequality the fact that $|ab-cd|\leq|a-c|(b+d)$ for $a,b,c,d>0$. In the third inequality, we used that $1-e^{-x}\leq x$ for all $x>0$. Finally, the fourth inequality follows, for instance, by \citep[pp. 108]{liu2012}.

The upper bound for the term given in (\ref{term_I}) follows exactly as discussed in \cite{liu2012}: choose $\epsilon'>0$ and $\eta>0$ small enough such as $\epsilon'+\dfrac{8\eta}{1-\|{\bf A}\|_p}<\epsilon$ and $\|{\bf x}-{\bf z}\|_p<\eta$ implies $|w({\bf x})-w({\bf z})|<\epsilon'$, with $p\in[1,\infty]$. In this part, we are using that $\|{\bf A}\|_p<1$ for some $p\in[1,\infty]$.

By combining the above results, we obtain that $\left|E(w(\boldsymbol\xi_1)|\boldsymbol\xi_0={\bf x})-E(w(\boldsymbol\xi_1)|\boldsymbol\xi_0={\bf z})\right|<\epsilon'+2\|{\bf x}-{\bf z}\|_p$. For general $k\geq 2$, the result follows by using mathematical induction, exactly as done in Chapter 4 of \cite{liu2012}, and therefore it is omitted. With the e-chain property established for the bivariate process $\{\boldsymbol\xi_t\}_{t\in\mathbb Z}$, we show part (b) by following similar arguments of the proofs of Proposition 4.2.1(b) and Proposition 4.3.1 of \cite{liu2012} to establish the geometric moment contraction (GMC) of $\{\boldsymbol\xi_t\}_{t\in\mathbb Z}$.

Denote by $(Y_t,\boldsymbol\xi_t)$ and $(Y_t^{(0)},\boldsymbol\xi_t^{(0)})$ the process under the initial value $\boldsymbol\xi_0=\boldsymbol\xi=(\lambda,\phi)^\top$ and $\boldsymbol\xi_0=\boldsymbol\xi^{(0)}=(\lambda^{(0)},\phi^{(0)})^\top$, respectively. It follows that
\begin{eqnarray*}
E\left(\|\boldsymbol\xi_1-\boldsymbol\xi_1^{(0)}\|_p\right)&=&E\left(\|(\boldsymbol\tau+{\bf A}\boldsymbol\xi+{\bf B}(Y_0,Y_0)^\top)-(\boldsymbol\tau+{\bf A}\boldsymbol\xi^{(0)}+{\bf B}(Y_0^{(0)},Y_0^{(0)})^\top)\|_p\right)\\
&\leq&\|{\bf A}\|_p\|\boldsymbol\xi-\boldsymbol\xi^{(0)}\|_p+E\left(\|{\bf B}(Y_0-Y_0^{(0)},Y_0-Y_0^{(0)})^\top\|_p\right)\\
&=&\|{\bf A}\|_p\|\boldsymbol\xi-\boldsymbol\xi^{(0)}\|_p+\|{\bf B}\|_pE\left(|Y_0-Y_0^{(0)}|\right).
\end{eqnarray*}
Further, we have that $E(|Y_0-Y_0^{(0)}|)=|\lambda-\lambda^{(0)}|$; for instance, see proof of Thm 1 by \cite{chrfok2014}.
Hence,
\begin{eqnarray*}
E\left(\|\boldsymbol\xi_1-\boldsymbol\xi_1^{(0)}\|_p\right)&\leq&\|{\bf A}\|_p\|\boldsymbol\xi-\boldsymbol\xi^{(0)}\|_p+\|{\bf B}\|_p|\lambda-\lambda^{(0)}|\\
&\leq&\|{\bf A}\|_p\|\boldsymbol\xi-\boldsymbol\xi^{(0)}\|_p+\|{\bf B}\|_p\|\boldsymbol\xi-\boldsymbol\xi^{(0)}\|_1\\
&\leq&\|{\bf A}\|_p\|\boldsymbol\xi-\boldsymbol\xi^{(0)}\|_p+2^{1-1/p}\|{\bf B}\|_p\|\boldsymbol\xi-\boldsymbol\xi^{(0)}\|_p\\
&\leq&(\|{\bf A}\|_p+2^{1-1/p}\|{\bf B}\|_p)\|\boldsymbol\xi-\boldsymbol\xi^{(0)}\|_p,
\end{eqnarray*}
where we have used that $\|\boldsymbol\xi-\boldsymbol\xi^{(0)}\|_1\leq2^{1-1/p}\|\boldsymbol\xi-\boldsymbol\xi^{(0)}\|_p$. Since $\|{\bf A}\|_p+2^{1-1/p}\|{\bf B}\|_p<1$, the GMC is satisfied, which concludes the proof. $\square$

\subsection{Proof of Theorem \ref{moments_P-INGARCH}}

Let $\{Y_t\}_{t\in\mathbb N}$ be a first-order negative binomial tv-DINGARCH process. Under condition $\|{\bf A}\|_1+\|{\bf B}\|_1<1$, Theorem \ref{Thm_StatErg} holds (for $p=1$) and gives us that $\{ \boldsymbol{\xi} _t\}_{t\in\mathbb Z}$ has a stationary and ergodic solution. Moreover, this process is a Markov chain. Therefore, we need to show  that Tweedie's criterion (for instance, see Chapter 15 from  \cite{meytwe1993}) is satisfied for an appropriate test function to prove the finiteness of arbitrary moments.
	
	Consider the test function $V(x,z)=1+x^k+z^k$, for $x,z>0$ and arbitrary $k\in\mathbb N$. It follows that
	\begin{eqnarray}\label{tweedie}
		E\left(V(\lambda_t,\phi_t)|\lambda_{t-1}=\lambda, \phi_{t-1}=\phi\right)&=&1+E\left(\lambda_t^k|\lambda_{t-1}=\lambda,\phi_{t-1}=\phi\right)+E\left(\phi_t^k|\lambda_{t-1}=\lambda,\phi_{t-1}=\phi\right)\nonumber\\
		&=&1+E\left((\beta_0+\beta_1Y+\beta_2\lambda)^k\right)+E\left((\alpha_0+\alpha_1Y+\alpha_2\phi)^k\right),			
	\end{eqnarray}	
	where $Y\sim\mbox{NB}(\lambda,\phi)$. 	 From proof of Lemma (A.1) from \cite{chrfok2014}, we have that the $d$-th moment of $Y\sim\mbox{NB}(\lambda,\phi)$, for $d\geq 2$, satisfies the recursive equation
	\begin{eqnarray}\label{NB_moments}
		E(Y^d)=\lambda\left\{1+\sum_{j=1}^{d-1}\left[\binom{d-1}{j}+\dfrac{1}{\phi} \binom{d-1}{j-1}\right]E(Y^j)\right\},
	\end{eqnarray}	
	with $E(Y)=\lambda$, which yields that $E(Y^d)=\lambda^d(1+1/\phi)+O(\lambda^{d-1})$ as $\lambda\rightarrow\infty$. Using this result, it follows that
	\begin{eqnarray*} 
		E\left((\beta_0+\beta_1Y+\beta_2\lambda)^k\right)&=&\sum_{j=0}^k\binom{k}{j}(\beta_0+\beta_2\lambda)^{k-j}\beta_1^jE(Y^j)=\sum_{j=0}^k\binom{k}{j}(\beta_2\lambda)^{k-j}\beta_1^j\lambda^j+O(\lambda^{k-1})\\
		&=&\lambda^k(\beta_1+\beta_2)^k+O(\lambda^{k-1}).
	\end{eqnarray*}
	
	Similarly, we obtain that
	\begin{eqnarray*} 
		E\left((\alpha_0+\alpha_1Y+\alpha_2\phi)^k\right)=(\alpha_1\lambda+\alpha_2\phi)^k+O(\lambda^{k-1}).
	\end{eqnarray*}
Hence,  (\ref{tweedie})  becomes 
	\begin{eqnarray}\label{tweedie2}
		E\left(V(\lambda_t,\phi_t)|\lambda_{t-1}=\lambda, \phi_{t-1}=\phi\right)=V(\lambda,\phi)\dfrac{1+\lambda^k(\beta_1+\beta_2)^k+(\alpha_1\lambda+\alpha_2\phi)^k+O(\lambda^{k-1})}{1+\lambda^k+\phi^k}.			
	\end{eqnarray}	
	
	We will now analyze the behavior of the terms multiplying $V(\lambda,\phi)$ in (\ref{tweedie2}) for large values of $\lambda$ and $\phi$. First, note that $\lambda^*=\dfrac{\beta_0}{1-\beta_2}$ and $\phi^*=\dfrac{\alpha_0}{1-\alpha_2}$ are the smallest reachable points of $\lambda_t$ and $\phi_t$, respectively, for all $t$. Hence, we consider $\lambda\in[\lambda^*,\infty)$ and $\phi\in[\phi^*,\infty)$, both bounded away from zero. Consider the polar coordinates transformation $\lambda=r\cos\delta$ and $\phi=r\sin\delta$, for $r>0$ and $\delta\in(0,\pi/2)$ such that $\lambda\geq\lambda^*$ and $\phi\geq\phi^*$. In terms of polar coordinates, we have that 
	\begin{eqnarray*}
		\dfrac{1+\lambda^k(\beta_1+\beta_2)^k+(\alpha_1\lambda+\alpha_2\phi)^k}{1+\lambda^k+\phi^k}&=&\dfrac{1+r^k[(\beta_1+\beta_2)^k\cos^k\delta+(\alpha_1\cos\delta+\alpha_2\sin\delta)^k]}{1+r^k(\cos^k\delta+\sin^k\delta)}\\
		&\longrightarrow&\dfrac{(\beta_1+\beta_2)^k\cos^k\delta+(\alpha_1\cos\delta+\alpha_2\sin\delta)^k}{\cos^k\delta+\sin^k\delta}\\
		&\leq&\dfrac{(\beta_1+\beta_2)^k\cos^k\delta+\max\{\cos\delta,\sin\delta\}(\alpha_1+\alpha_2)^k}{\cos^k\delta+\sin^k\delta}\\
		&\leq&(\beta_1+\beta_2)^k+(\alpha_1+\alpha_2)^k\\
		&\leq&\beta_1+\beta_2+\alpha_1+\alpha_2\\
		&=&\|{\bf A}\|_1+\|{\bf B}\|_1<1,
	\end{eqnarray*}	
	as $r\rightarrow\infty$, where the last inequality follows from lemma's assumption implying that $\beta_1+\beta_2<1$ and $\alpha_1+\alpha_2<1$. Furthermore, the remaining term converges to 0 as $r\rightarrow\infty$ because  
	$\dfrac{O(\lambda^{k-1})}{1+\lambda^k+\phi^k}=O(r^{k-1})/O(r^{k})=O(r^{-1})$.
	
	Therefore, there exist real constants $\kappa_1\in(0,1)$, $\kappa_2>0$, and $L>0$ such that (\ref{tweedie}) is bounded from above as follows:
	\begin{eqnarray*}
		E\left(V(\lambda_t,\phi_t)|\lambda_{t-1}=\lambda, \phi_{t-1}=\phi\right)\leq(1-\kappa_1) V(\lambda,\phi)+\kappa_2I\{(\lambda,\phi)\in G\},
	\end{eqnarray*}	
	where $G=\{\lambda\geq\lambda^*,\phi\geq\phi^*:\, \lambda=r\cos\delta,\, \phi=r\sin\delta,\, 0<r<L,\, \delta\in(0,\pi/2)\}$.
	
	In other words, Tweedie's criterion is satisfied and we conclude that the $k$-th moment of $\lambda_t$ and $\phi_t$ are finite for arbitrary $k\in\mathbb N$. Now, using the fact that $Y_t|\mathcal F_{t-1}\sim\mbox{NB}(\lambda_t,\phi_t)$ and Eq. (\ref{NB_moments}), we have that the $k$-th conditional moment of $Y_t$ given $\mathcal F_{t-1}$ is a polynomial on $\lambda_t$ and $1/\phi_t$; note that $1/\phi_t\leq1/\alpha_0$ for all $t$. Since the $l$-th moment of $\lambda_t$ is finite for arbitrary $l\in\mathbb N$, we obtain that the unconditional $k$-th moment of $Y_t$ is finite for all $k\geq1$ as well. $\square$

\begin{Lemma} \label{lemma:starting value}
Under the conditions of Theorem  \ref{asymptotics}
$$
\sup\limits_{\boldsymbol{\theta}\in\mathcal{B}}\Big|\frac{1}{n}\ell(\boldsymbol{\theta})-\frac{1}{n}\tilde{\ell}(\boldsymbol{\theta})\Big|\overset{a.s.}
\longrightarrow 0~~~~\mbox{as}~n\rightarrow\infty
$$
following  the notation from Subsection \ref{subsec:estimation}.
\end{Lemma}

\begin{proof}
		Let  $\boldsymbol\xi_t=(\lambda_t,\phi_t)^\top$ denote the stationary and ergodic solution of model \eqref{DINGARCH}. In addition, let 
		$\tilde{\boldsymbol{\xi}}_t=( \tilde{\lambda}_t,\tilde{\phi}_t)^\top$ the process obtained with some starting value 
		$\tilde{\boldsymbol{\xi}}_0=( \tilde{\lambda}_0,\tilde{\phi}_0)^\top$. Then, arguing as in \citet[Lemma A.2]{chrfok2014} and using 
		that   $\|{\bf A}\|_1<1$. we obtain that 
		$$
		\sup\limits_{\boldsymbol{\theta}\in\mathcal{B}}\|\boldsymbol{\xi}_t-\tilde{\boldsymbol{\xi}}_t\|_{1} \leq C \|{\bf A}\|^t_1,~~~ \forall ~t,
		$$
		almost surely from the compactness of $\mathcal{B}$, for some constant $C > 0$. Now,  consider the differences $\ell_{t}(\boldsymbol{\theta})-\tilde{\ell}_{t}(\boldsymbol{\theta})$
		following  the notation introduced by \eqref{loglikelihood} and \eqref{loglikelihood using the starting value}.  By grouping terms we obtain the following inequalities using repeatedly the facts that for any $x,y>0$ it holds that $|\log (x/y)|\leq |x-y|/\min(x,y)$, $ \lambda_{t} > \beta_{0}$, $\phi_{t} > \alpha_0$ and $\log x < x-1$, for all $x >0$.
		\begin{align*}
			\bigg|
			y_t \Bigl[\log\lambda_t-\log(\lambda_t+\phi_t) \Bigr]- y_t\Bigr[
			\log\tilde{\lambda}_{t}-\log(\tilde{\lambda}_t+\tilde{\phi}_t) \Bigl]\bigg|& \leq  2y_{t} \frac{\| \xi_t- \tilde{\xi}_{t}\|_{1}}{\beta_{0}},
		\end{align*}
		\begin{align*}
			\bigg|
			\phi_t \Bigl[\log\phi_t-\log(\lambda_t+\phi_t) \Bigr]- \tilde{\phi}_t\Bigr[
			\log\tilde{\phi}_{t}-\log(\tilde{\lambda}_t+\tilde{\phi}_t) \Bigl]\bigg|& \leq 
			\| \xi_{t}- \tilde{\xi}_t \|_{1} \Bigl\{ \Bigl[   (\frac{2}{\alpha_0}+1) +\phi_{t} + \lambda_{t}+2 \| \xi_{t}-\tilde{\xi}_t\|_{1} \Bigr]\Bigr\},
		\end{align*}
		\begin{align*}
			\Bigr|	\log\Gamma(y_t+\phi_t)- 	\log\Gamma(y_t+\tilde{\phi}_t) \Bigl| \leq (\frac{2}{\alpha_0}+1)(y_t+\phi_t +1)
			\| \xi_{t}-\tilde{\xi}_t\|_{1},
		\end{align*}
		\begin{align*}
			\Bigr|	\log\Gamma(\phi_t)- 	\log\Gamma(\tilde{\phi}_t) \Bigl| \leq (\frac{2}{\alpha_0}+1)(\phi_t +1)
			\| \xi_{t}-\tilde{\xi}_t\|_{1}.
		\end{align*}
		For the last two inequalities, we employ Stirling's approximation to the Gamma function. In conclusion 
		\begin{eqnarray*}
			\sup\limits_{\boldsymbol{\theta}\in\mathcal{B}}\Big|\frac{1}{n}\ell(\boldsymbol{\theta})-\frac{1}{n}\tilde{\ell}(\boldsymbol{\theta})\Big|	
			&\leq& \frac{C}{n}\sum_{t=1}^{n} \| \boldsymbol{A} \|^{t}_1 W_t,
		\end{eqnarray*}
		where $W_t$ is a positive linear combination of $(Y_t, \lambda_t, \phi_t)$. Theorem  \ref{moments_P-INGARCH} implies that  $E|W_{t}|^k < \infty$, for all $k \in \mathbb{N}$. 
		\noindent
But using Markov$^\prime$s inequality for every $\epsilon>0$ we have that
$$
\sum_{t=1}^{\infty}P\Big(\| \boldsymbol{A} \|^{t}_1 W_t >\epsilon\Big) \leq \displaystyle 
\sum_{t=1}^{\infty}\frac{\| \boldsymbol{A} \|^{kt}_1 E(W_t^k)}{\epsilon^k}<\infty.
$$
\noindent
Hence, $ \| \boldsymbol{A} \|^{t}_1 W_t\overset{a.s.}\longrightarrow 0$, as $t\rightarrow \infty$ and therefore the proof is concluded.
\end{proof}

\subsection{Proof of Lemma \ref{lemma:integrability}}

We have that $|S_{1t}|\leq|y_t/\lambda_t-1|\leq y_t/\beta_0+1$. Furthermore, $\partial\lambda_t/\partial\beta_i$ is a linear combination of $y_{t-1},\ldots,y_1$ and $\lambda_1$, for $i=0,1,2$. Using these results and under condition $\|{\bf A}\|_1+\|{\bf B}\|_1<1$, Theorem \ref{moments_P-INGARCH} gives that the score function associated with $\beta_0,\beta_1$ and $\beta_2$ is integrable. We now study the score function associated with $\alpha_0$, $\alpha_1$ and $\alpha_2$. Note that $\left|-\dfrac{y_t-\lambda_t}{\lambda_t+\phi_t}\right|\leq \dfrac{y_t}{\lambda_t+\phi_t}+1\leq\dfrac{y_t}{2}\left(\dfrac{1}{\beta_0}+\dfrac{1}{\alpha_0}\right)+1$, where we have used that fact that  $\dfrac{1}{a+b}\leq\dfrac{1}{2}\left(\dfrac{1}{a}+\dfrac{1}{b}\right)$ for $a,b>0$. Also, by using that $\dfrac{1}{2x}<\log x-\Psi(x)<\dfrac{1}{x}$ for $x>0$ \citep{alz1997}, we obtain that $0<\log\phi_t-\Psi(\phi_t)<\phi_t^{-1}<\alpha_0^{-1}$. Similarly, it follows that $|\Psi(\lambda_t+\phi_t)-\log(\lambda_t+\phi_t)|<(\lambda_t+\phi_t)^{-1}\leq\dfrac{1}{2}\left(\dfrac{1}{\beta_0}+\dfrac{1}{\alpha_0}\right)$. Therefore, $|S_{2t}|$ is bounded above by a linear combination of $y_t$. Moreover, $\dfrac{\partial\phi_t}{\partial\alpha_j}$ is a linear combination of $y_{t-1},\ldots,y_1$ and $\phi_1$, for $j=0,1,2$. Again, an application of  Theorem \ref{moments_P-INGARCH} gives that the score function associated with $\alpha_0$, $\alpha_1$ and $\alpha_2$ is also integrable.

\begin{Lemma}\label{third_derivat_bound}
Let $\mathcal B$ as defined in (\ref{B-set}). Under the linear negative binomial tv-DINGARCH(1,1,1,1) model with $\|{\bf A}\|_1+\|{\bf B}\|_1<1$, there exists a sequence of random variables $\{M_n\}_{n\in\mathbb N}$ such that $$\max_{i,j,k=1,2,\ldots,6}\sup_{\boldsymbol\theta\in\mathcal B}\bigg|\dfrac{1}{n}\dfrac{\partial^3\ell(\boldsymbol\theta)}{\partial\theta_i\partial\theta_j\partial\theta_k}\bigg|\leq M_n,$$ where $M_n\stackrel{p}{\longrightarrow}m$ as $n\rightarrow\infty$, with $m<\infty$ and $\boldsymbol\theta=(\theta_1,\theta_2,\theta_3,\theta_4,\theta_5,\theta_6)^\top=(\beta_0,\beta_1,\beta_2,\alpha_0,\alpha_1,\alpha_2)^\top$. 
\end{Lemma}

\begin{proof}
We will show that the result holds for the third derivative of $\ell(\cdot)$ with respect to $\beta_1$ and also $\alpha_1$. The remaining cases are similar and therefore are omitted. 

The third derivative of $\ell_t(\boldsymbol\theta)$ with respect to $\beta_1$ is given by $\dfrac{\partial^3\ell_t(\boldsymbol\theta)}{\partial\beta_1^3}=\dfrac{\partial^2 S_{1t}}{\partial\phi_t^2}\left(\dfrac{\partial\lambda_t}{\partial\beta_1}\right)^3$, where
\begin{eqnarray*}
\dfrac{\partial^2 S_{1t}}{\partial\phi_t^2}=\dfrac{2\phi_t}{\lambda_t^2(\lambda_t+\phi_t)^2}\left\{\dfrac{[y_t(2\lambda_t+\phi_t)-\lambda_t^2](2\lambda_t+\phi_t)}{\lambda_t(\lambda_t+\phi_t)}-(y_t-\lambda_t)\right\}
\end{eqnarray*}
and 
\begin{eqnarray*}
	\dfrac{\partial\lambda_t}{\partial\beta_1}=\sum_{j=1}^{t-1}\beta_2^{j-1}y_{t-j}.
\end{eqnarray*}

It is straightforward to obtain that there are constants $c_1,c_2>0$ that do not depend on $\boldsymbol\theta$ such that $\left|\dfrac{\partial^3\ell_t(\boldsymbol\theta)}{\partial\beta_1^3}\right|\leq c_2y_t+c_1$. Also, $\left|\dfrac{\partial\lambda_t}{\partial\beta_1}\right|\leq\displaystyle\sum_{j=1}^{t-1}\beta_{2,up}^{j-1}y_{t-j}\equiv \mu_{1t}$. Therefore, $\left|\dfrac{\partial^3\ell_t(\boldsymbol\theta)}{\partial\beta_1^3}\right|\leq c_2y_t\mu_{1t}^3+c_1\mu_{1t}^3\equiv M_{1t}$, 
where $M_{1t}$ does not depend on $\boldsymbol\theta$. Under the assumption $\|{\bf A}\|_1+\|{\bf B}\|_1<1$ and using Theorem \ref{moments_P-INGARCH}, the Law of Large Numbers for stationary and ergodic sequences can be applied to obtain that $n^{-1}\displaystyle\sum_{t=2}^n M_{1t}\stackrel{p}{\longrightarrow}m_1$  as $n\rightarrow\infty$, where $m_1$ is a finite constant, as similarly argued by \cite{foketal2009} (proof of Lemma 3.4 of their Supplementary Material).

The case for $\alpha_1$ is more involved. We have that $\dfrac{\partial^3\ell_t(\boldsymbol\theta)}{\partial\alpha_1^3}=\dfrac{\partial^2S_{2,t}}{\partial\phi_t^2}\left(\dfrac{\partial\phi_t}{\partial\alpha_1}\right)^3$, where 
\begin{eqnarray}\label{d2Sdphi2}
\dfrac{\partial^2S_{2,t}}{\partial\phi_t^2}&=&-2\dfrac{y_t-\lambda_t}{(\lambda_t+\phi_t)^3}+\dfrac{\lambda_t^2+2\lambda_t\phi_t}{\phi_t^2(\lambda_t+\phi_t)^2}+\Psi''(y_t+\phi_t)-\Psi''(\phi_t)
\end{eqnarray}
and
\begin{eqnarray*}
\dfrac{\partial\phi_t}{\partial\alpha_1}=\displaystyle\sum_{j=1}^{t-1}\alpha_2^{j-1}y_{t-j}.	
\end{eqnarray*}	

Inequality (15) from \cite{guoqi2013} gives that $|\Psi''(x)|<\dfrac{1}{x}+\dfrac{2}{x^3}$, for $x>0$. Hence, it follows that $|\Psi''(y_t+\phi_t)|<\dfrac{1}{(y_t+\phi_t)^2}+\dfrac{2}{(y_t+\phi_t)^3}\leq\dfrac{1}{\phi_t^2}+\dfrac{2}{\phi_t^3}\leq\dfrac{1}{\alpha_0^2}+\dfrac{2}{\alpha_0^3}\leq\dfrac{1}{\alpha_{0,low}^2}+\dfrac{2}{\alpha_{0,low}^3}$. Similarly, $|\Psi''(\phi_t)|<\dfrac{1}{\alpha_{0,low}^2}+\dfrac{2}{\alpha_{0,low}^3}$.
The modulus of the remaining terms in (\ref{d2Sdphi2}) are bounded above (as done for the third derivative involving $\beta_1$) by a linear function of $y_t$ with positive coefficients do not depend on $\boldsymbol\theta$. Moreover, $\left|\dfrac{\partial\phi_t}{\partial\alpha_1}\right|\leq\displaystyle\sum_{j=1}^{t-1}\alpha_{2,up}^{j-1}y_{t-j}\equiv \mu_{2t}$.

From the above results, we conclude that there exist constants $c_3>0$ and $c_4>0$ that do not depend on $\boldsymbol\theta$ such that 
$\left|\dfrac{\partial^3\ell_t(\boldsymbol\theta)}{\partial\alpha_1^3}\right|\leq c_4 y_t\mu_{2t}^3+c_3\mu_{2t}^3\equiv M_{2t}$.
Now, by arguing exactly as done for the case $\beta_1$ above, we obtain that $n^{-1}\sum_{t=2}^n M_{2t}\stackrel{p}{\longrightarrow} m_2$ when $n\rightarrow\infty$, where $m_2$ is a finite constant.
\end{proof}

\subsection{Proof of Theorem \ref{asymptotics}}

We will show that the Conditions (A1), (A2), and (A3) from Lemma 3.1 by \cite{jenrah2004} (see also \cite{foketal2009} and \cite{chrfok2014}) are satisfied, which yield the theorem. As argued before for the general case, we use Proposition \ref{martingale_diff} and the Central Limit Theorem for martingale difference to establish the weak convergence involving the score function in (\ref{score_asympt}). The existence and finiteness of the matrix $\boldsymbol\Omega_1(\boldsymbol\theta)$ follow from Theorem \ref{moments_P-INGARCH} that provides the finiteness of the moments of all orders for the linear NB tv-DINARCH process under the assumption that $\|{\bf A}\|_1+\|{\bf B}\|_1<1$. Therefore, Condition (A1) holds.
	
Moreover, we obtain from Theorem \ref{Thm_StatErg} that $\{Y_t\}_{t\in\mathbb Z}$ is a stationary and ergodic sequence. Then, the Law of Large Numbers for stationary and ergodic sequences gives us that (\ref{hessian_asympt}) holds true, with the existence and finiteness of the matrix $\boldsymbol\Omega_2(\boldsymbol\theta)$ being ensured again by Theorem \ref{moments_P-INGARCH}.

We prove next that  $\boldsymbol\Omega_1(\boldsymbol\theta)$ is non-singular.  
Suppose that it is singular. Then there exist vectors $\nu_1$, $\nu_{2}$ both in $\mathbb{R}^3$ such that the six-dimensional vector $\nu=(\nu_1^\top, \nu_2^\top)^\top$ satisfies 
$
\nu_1^\top {\partial\lambda_{t}}/{\partial\boldsymbol{\theta}_{1}}+
\nu_2^\top {\partial\phi_{t}}/{\partial\boldsymbol{\theta}_{2}}=0
$, a.s for $t \geq 1$. But both ${\partial\lambda_{t}}/{\partial\boldsymbol{\theta}_{1}}$ and 
$ {\partial\phi_{t}}/{\partial\boldsymbol{\theta}_{2}}$ are stationary so this holds for all $t$. However, 
\begin{align*}
\frac{\partial\lambda_{t}}{\partial\boldsymbol{\theta}_{1}} & = 
(1, Y_{t-1} ,\lambda_{t-1})^\top
+ 
\beta_2  \frac{\partial\lambda_{t-1}}{\partial\boldsymbol{\theta}_{1}}, \\ 
\frac{\partial\phi_{t}}{\partial\boldsymbol{\theta}_{2}} & = 
(1, Y_{t-1} ,\phi_{t-1})^\top
+ 
\alpha_2  \frac{\partial\phi_{t-1}}{\partial\boldsymbol{\theta}_{2}}.
\end{align*} 
The matrix $\boldsymbol\Omega_1(\boldsymbol\theta)$ is block-diagonal. Therefore it is  singular if and only if  $\nu_1^\top (1, Y_{t-1} ,\lambda_{t-1})^\top + \nu_2^\top (1, Y_{t-1} ,\phi_{t-1})^\top =0$ a.s for all $t \geq 1$.  
But this cannot hold because of the constant term 1 and the fact that 
$Y_{t}$,  $\lambda_{t}$ and $\phi_{t}$ 
are non-zero, for some $t$. 

That is, Condition (A2)  from Lemma 3.1 by \cite{jenrah2004} is satisfied. Finally, Condition (A3) has been established in Lemma \ref{third_derivat_bound}.

\section*{Acknowledgments}
We would like to thank the Editor and two reviewers for their careful reading and constructive comments that improved the presentation. 
H. Ombao would like to acknowledge support from the KAUST Research Fund. L.S.C. Piancastelli wishes to acknowledge the financial support from the Science Foundation Ireland.

\section*{Declarations}
\noindent {\bf Conflict of interest:} The authors have no conflicts of interest to declare.

\end{document}